\documentclass[a4paper,twocolumn,11pt,allowtoday]{quantumarticle}
\pdfoutput=1
\usepackage[utf8]{inputenc}
\usepackage[english]{babel}
\usepackage[T1]{fontenc}
\usepackage{hyperref}

\usepackage[square,numbers,sort&compress]{natbib}
\usepackage{graphicx}

\usepackage{amsmath,amsfonts,amssymb,amsthm,bm}
\usepackage{booktabs}

\usepackage{mathrsfs}

\usepackage[dvipsnames]{xcolor}

\usepackage{csquotes}
\MakeOuterQuote{"}
\usepackage{mathtools,mleftright}
\mleftright
\def\identity{\leavevmode\hbox{\small1\kern-3.8pt\normalsize1}}
\newtheorem{theorem}{Theorem}
\newtheorem{lemma}{Lemma}

\newtheorem{corollary}{Corollary}
\newtheorem{proposition}{Proposition}

\theoremstyle{plain}
\newtheorem{definition}{Definition}

\theoremstyle{remark}

\newcommand{\Cl}{\mathrm{Cl}}

\newcommand{\haar}{\mathrm{Haar}}

\newcommand{\rmd}{\mathrm{d}}
\newcommand{\rme}{\operatorname{e}}
\newcommand{\rmi}{\mathrm{i}}
\newcommand{\rmG}{\mathrm{G}}
\newcommand{\rmS}{\mathrm{S}}
\newcommand{\rmR}{\mathrm{R}}

\newcommand{\rmU}{\mathrm{U}}
\newcommand{\SWAP}{\mathrm{SWAP}}

\newcommand{\bfa}{\mathbf{a}}
\newcommand{\bfm}{\mathbf{m}}
\newcommand{\bfn}{\mathbf{n}}
\newcommand{\bfb}{\mathbf{b}}
\newcommand{\bfc}{\mathbf{c}}
\newcommand{\bfd}{\mathbf{d}}
\newcommand{\bfx}{\mathbf{x}}
\newcommand{\bfy}{\mathbf{y}}
\newcommand{\caE}{\mathcal{E}}
\newcommand{\caH}{\mathcal{H}}
\newcommand{\caO}{\mathcal{O}}
\newcommand{\caP}{\mathcal{P}}
\newcommand{\bbE}{\mathbb{E}}

\newcommand{\scrN}{\mathscr{N}}
\newcommand{\scrL}{\mathscr{L}}
\newcommand{\scrM}{\mathscr{M}}

\newcommand{\tr}{\operatorname{Tr}}

\newcommand{\lref}[1]{Lemma~\ref{#1}}
\newcommand{\Lsref}[2]{Lemmas~\ref{#1} and \ref{#2}}
\newcommand{\thref}[1]{Theorem~\ref{#1}}
\newcommand{\pref}[1]{Proposition~\ref{#1}}
\newcommand{\coref}[1]{Corollary~\ref{#1}}
\def\eqref#1{\textup{(\ref{#1})}}
\newcommand{\eref}[1]{Eq.~\textup{(\ref{#1})}}
\newcommand{\eqsref}[2]{Eqs.~(\ref{#1}) and (\ref{#2})}
\newcommand{\Eref}[1]{Equation~\textup{(\ref{#1})}}
\newcommand{\Eqsref}[2]{Equations~(\ref{#1}) and (\ref{#2})}
\newcommand{\tref}[1]{Table~\ref{#1}}
\newcommand{\sref}[1]{Sec.~\ref{#1}}
\newcommand{\fref}[1]{Fig.~\ref{#1}}
\newcommand{\Fref}[1]{Figure~\ref{#1}}
\newcommand{\aref}[1]{Appendix~\ref{#1}}

\def\<{\langle}  
\def\>{\rangle}
\newcommand{\rcite}[1]{Ref.~\cite{#1}}
\newcommand{\rscite}[1]{Refs.~\cite{#1}}

\begin{document}

\title{Unitary designs from perturbed time evolutions of a chaotic Hamiltonian}
\author{Zhongyi Yang}
\affiliation{International Center for Quantum Materials, School of Physics, Peking University, Beijing 100871, China}
\author{Biao Wu}
\email{wubiao@pku.edu.cn}
\affiliation{International Center for Quantum Materials, School of Physics, Peking University, Beijing 100871, China}
\affiliation{Wilczek Quantum Center, Shanghai Institute for Advanced Studies, University of Science and Technology of China, Shanghai 201315, China}
\affiliation{Hefei National Laboratory, Hefei 230088, China}
\date{\today}

\begin{abstract}
Unitary designs provide efficient substitutes for Haar-random unitaries in
quantum information processing, randomized measurements, and many-body
quantum dynamics. We propose a protocol based on time evolutions of a
single chaotic Hamiltonian interleaved with unitary perturbations to generate
unitary designs. We derive the frame potential of the resulting ensemble in
terms of those of the intermediate ensemble. The intermediate ensemble need
not itself be Haar random or exhibit Haar-level frame potentials, provided
that its trace moments and frame-potential growth are sufficiently suppressed.
The nontrivial Pauli set and the Clifford ensemble are simple examples,
whereas ensembles whose cardinality does not grow with the Hilbert-space
dimension generally fail. We further establish a frame-potential correspondence between
multi-Hamiltonian temporal protocols and one-Hamiltonian protocols
interleaved with repeated fixed trace-suppressed perturbations, showing
that two perturbations suffice to recover Haar-level frame potentials.
For local Hamiltonians, we derive a lower bound on the space--depth
control resource, while more generally we establish a universal
finite-window bound that constrains the required temporal sampling
range. Together, these results identify complementary spatial and
temporal resource requirements for perturbation-assisted randomness
generation.
\end{abstract}

\maketitle

\section{Introduction}

Random unitaries underpin quantum chaos and thermalization \cite{RobertsStanford2015,cotler2017chaos,Liu2018entanglement,Kos2018many} and quantum-information tasks including state characterization \cite{HuangKuengPreskill2020,Eisert2020Cert,Elben2023}, randomized benchmarking \cite{Knill2008randomized,Helsen2022general,Nakata2021quantum}, and cryptography \cite{Portman2022security}. Unlike random-state generation through deep thermalization \cite{Cotler2023emccccergent,choi2023preparing,Zhang2025Holographic,Pilatowsky2024hilbert}, random-unitary generation requires randomness over the full Hilbert space. Because exact Haar sampling is exponentially costly, unitary \(k\)-designs, which reproduce Haar moments up to order \(k\), provide a practical substitute \cite{GrossAudenaertEisert2007,Dankert2009exact,Emerson2003pseudo}. A central challenge is to realize such designs using experimentally feasible dynamics.

In circuit-based architectures, approximate designs can be generated by deep local random circuits or shallow circuits with special structures \cite{BrandaoHarrowHorodecki2016,brandao2016efficient,Haferkamp2022,Nakata2021quantum,Schuster2024random}. However, these protocols typically rely on many independently sampled gates, so different realizations require distinct circuit configurations across sampling rounds. In analog quantum simulators and chaotic many-body systems, a more natural realization of randomness is through Hamiltonian time evolution \cite{Pilatowsky2023complete,Pilatowsky2024hilbert,cui2025random,mao2025random,Zhou2026realizing,Zhou2026three,sun2026unitary}. However, for a fixed Hamiltonian, the temporal orbit is constrained by the energy eigenbasis, so additional ingredients are needed to promote phase ergodicity into unitary ergodicity \cite{Roberts2017chaos,Mark2024}. Recent works showed that chaotic Hamiltonian dynamics can generate unitary designs either from three independently controlled chaotic Hamiltonians or from two chaotic Hamiltonians combined with an intermediate random Pauli operation \cite{Zhou2026three,sun2026unitary}. These results suggest a hierarchy of simplifications, illustrated in \fref{fig:procedure}, in which the amount of independent Hamiltonian control is progressively reduced. Here we ask whether unitary designs can be generated from a single fixed chaotic Hamiltonian, with randomness supplied only by sampled evolution times and intermediate unitary perturbations. This setting is particularly attractive for platforms like trapped ions and cold atoms where switching Hamiltonian settings is resource-intensive, making it a promising route to Hamiltonian-based random-unitary generation.

In this work, we show that a single chaotic Hamiltonian,
supplemented by intermediate unitary perturbations, is sufficient to
generate unitary designs. We consider an ensemble of two independently
timed evolutions under the same fixed Hamiltonian, separated by a
perturbation drawn from an intermediate ensemble. For typical chaotic
Hamiltonians satisfying additive nonresonance, we derive the $k$-th order
frame potential in terms of the lower-order frame potentials of the
intermediate ensemble. The intermediate ensemble need not be Haar
random or itself exhibit Haar-level frame potentials; it is sufficient
that its trace moments and frame-potential growth are sufficiently
suppressed relative to the maximal scaling $D^{2k}$, where $D$ is the
Hilbert-space dimension. Nontrivial Pauli and Clifford ensembles are
natural examples, while the mechanism is more general.

We further show that repeated applications of a fixed trace-suppressed
unitary can replace independent Hamiltonian controls at the level of
the frame potential. For a typical chaotic Hamiltonian, two applications
of the same perturbation suffice to attain the Haar frame potential
at any fixed order in the long-time, large-$D$ limit. The protocol
thus requires only one fixed Hamiltonian, independently sampled
evolution times, and a repeatable pulse such as a Pauli operation.

We also establish  two resource bounds,
requiring neither spectral nonresonance nor averaging over Hamiltonian
eigenbasis. For local Hamiltonians, we derive a lower bound on the
space--depth control resource, quantifying the tradeoff between
perturbation support and depth. For fixed pulse sequences, we establish
a universal finite-window bound arising from coarse-grained temporal
collisions, showing how increasing perturbation depth relaxes the
necessary sampling-window constraint. Numerical simulations for
Rydberg Hamiltonians illustrate the space--depth tradeoff, while
simulations for both Rydberg and GUE Hamiltonians exhibit finite-window
behavior consistent with the predicted depth dependence.

\section{Temporal ensembles and intermediate perturbations}

We consider an $n$-qubit Hilbert space $\caH$ of dimension $D=2^n$. Let $\nu$ be a probability distribution over the unitary group $U(D)$ on $\caH$. An ensemble $\nu$ is an exact unitary $k$-design if its first $k$ moments
coincide with those of the Haar ensemble.  In this work, we diagnose
unitary-design formation through the $k$-th frame potential
\begin{equation}
    F_{\nu}^{(k)}
    =
    \bbE_{U,V\sim\nu}
    \left|\tr(U^\dagger V)\right|^{2k}.
    \label{eq:frame_potential_intro}
\end{equation}
The Haar value $F_{\haar}^{(k)}=k!$ is the minimum over unitary ensembles and is attained precisely by exact unitary $k$-designs \cite{GrossAudenaertEisert2007,RoyScott2009}; hence $F_{\nu}^{(k)}\approx k!$ certifies approximate design formation.

To compare different Hamiltonian-based constructions, we denote by
$\caP_n=\{I,X,Y,Z\}^{\otimes n}$ the set of Hermitian Pauli strings
without phases, and define the nontrivial Pauli set $ \caP_{n}^*=\caP_n\setminus\{I^{\otimes n}\}$. The three temporal ensembles discussed in this work are
\begin{align}
    \caE_{\scrL}(H_1,H_2,H_3)
    &:=
    \left\{
        \rme^{-\rmi H_3 t_3}
        \rme^{-\rmi H_2 t_2}
        \rme^{-\rmi H_1 t_1}
    \right\},\notag\\
    \caE_{\scrM}(H_1,H_2,\caP_{n}^*)
    &:=
    \left\{
        \rme^{-\rmi H_2t_2}
        P
        \rme^{-\rmi H_1t_1}
    \right\},\notag\\
    \caE_{\scrN}(H,\caE_{\mathrm{int}})
    &:=
    \left\{
        \rme^{-\rmi Ht_2}
        U_{\mathrm{int}}
        \rme^{-\rmi Ht_1}
    \right\}.\label{eq:ABC_protocol_def}
\end{align}
Here $t_1,t_2,t_3\in[0,T]$ are sampled independently and uniformly.
$P$ and $U_{\mathrm{int}}$ are sampled from $\caP_n^*$ and $\caE_{\mathrm{int}}$ respectively.  The ensembles
$\caE_{\scrL}$ and $\caE_{\scrM}$ correspond respectively to the
three-Hamiltonian temporal protocol in \rcite{Zhou2026three} and the
Pauli-assisted two-Hamiltonian protocol in \rcite{sun2026unitary}.  

The
main object of this work is $\caE_{\scrN}$, where the same chaotic
Hamiltonian $H$ is used before and after the intermediate operation. 
\begin{figure}[!tbp]
    \centering
    \includegraphics[width=\columnwidth]{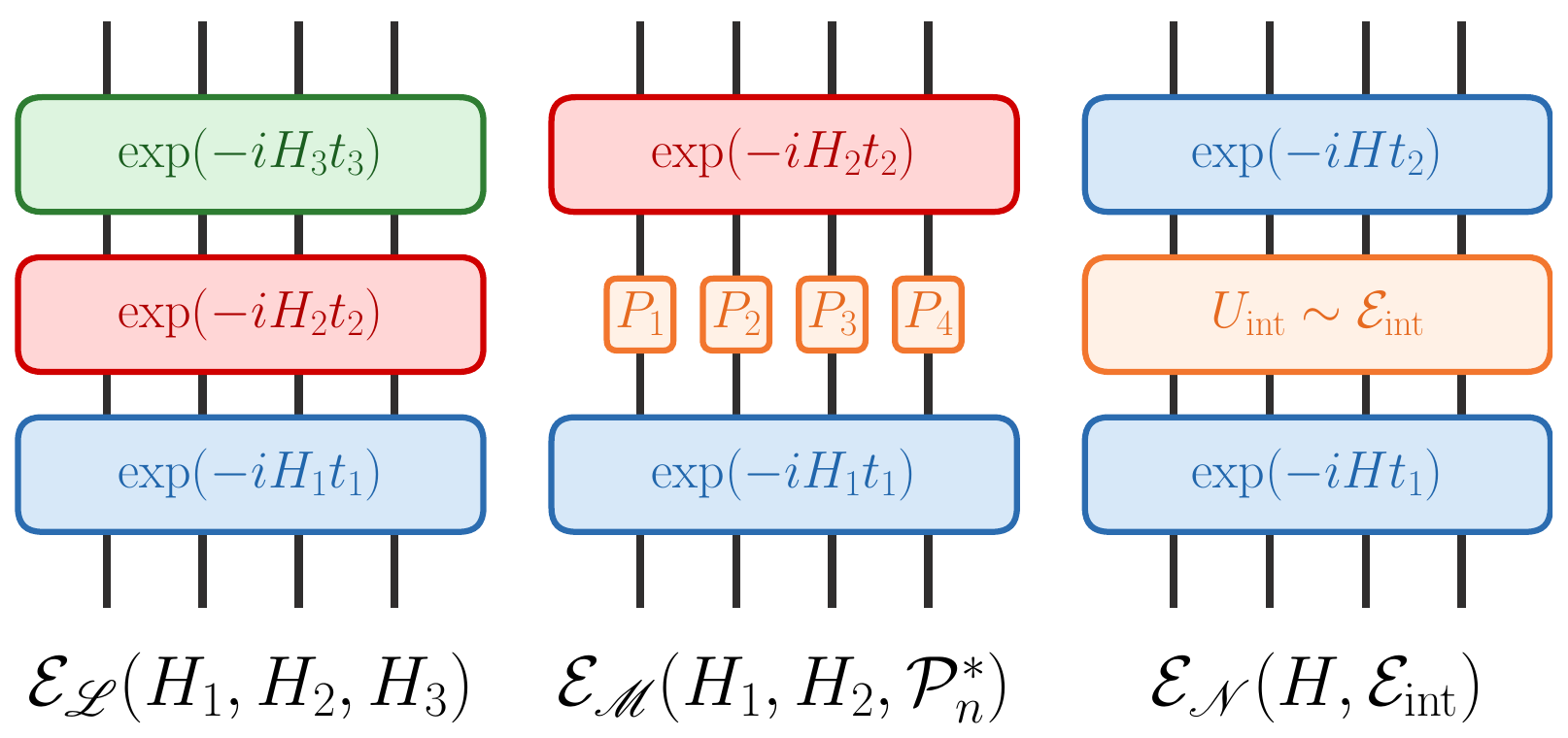}
\caption{Schematic of the three Hamiltonian-based unitary ensembles considered
in this work.  The protocol $\caE_{\scrL}$ uses three independently
controlled chaotic Hamiltonians; $\caE_{\scrM}$ uses two chaotic
Hamiltonians with an intermediate Pauli operation; and $\caE_{\scrN}$
uses one chaotic Hamiltonian together with an intermediate unitary
$U_{\mathrm{int}}\sim\caE_{\mathrm{int}}$.
}
    \label{fig:procedure}
\end{figure}
Explicitly, fix a chaotic Hamiltonian \(H\) and a time window \([0,T]\), and
consider two independent draws from the protocol,
\begin{equation}\label{eq:two_unitaries_convention}
   U=\rme^{-\rmi Ht_2}P\rme^{-\rmi Ht_1},
    \qquad
    V=\rme^{-\rmi Ht_2'}Q\rme^{-\rmi Ht_1'},
\end{equation}
where \(P,Q\sim\caE_{\mathrm{int}}\) are independent and all four evolution
times are sampled independently and uniformly from the window \([0,T]\). Let $\delta t_1=t_1-t_1',\,\delta t_2=t_2-t_2'.$
The frame potential then reads
\begin{align}
    F_{\caE_{\scrN}}^{(k)}&
    =\bbE_{U,V\sim \caE_{\scrN}}\left| \tr(UV^{\dagger})\right|^{2k} \notag\\
    &=
    \bbE_{\delta t_1,\delta t_2,P,Q}
    \left|
        \tr\!\left(
        Q^\dagger \rme^{-\rmi H\delta t_2}
        P \rme^{-\rmi H\delta t_1}
        \right)
    \right|^{2k}.
    \label{eq:quenched_frame_potential_def}
\end{align}
The intermediate ensemble enters this expression only through the independent
draws \(P\) and \(Q\). For a Hamiltonian \(H\) with eigenvalues
\(\{E_\alpha\}_{\alpha=1}^{D}\), the perfect time-filter limit
\(T\to\infty\) removes all oscillatory terms with nonzero frequency in
\eref{eq:quenched_frame_potential_def}. At order \(k\), the surviving
terms therefore satisfy
\(\sum_{r=1}^{k}(E_{a_r}-E_{b_r})=0\)
for two \(k\)-replica energy strings
\(\bfa=(a_1,\ldots,a_k)\) and
\(\bfb=(b_1,\ldots,b_k)\). To ensure that such zero-frequency
conditions arise only from trivial matching of the two energy
strings, we impose the following spectral assumption.

\begin{definition}[Additive nonresonance condition]
\label{def:k_nonresonance}
A Hamiltonian $H$ with eigenvalues $\{E_\alpha\}_{\alpha=1}^{D}$ satisfies
the $k$-th order additive nonresonance condition if, for any two
$k$-replica strings $\bfa=(a_1,\ldots,a_k)$ and
$\bfb=(b_1,\ldots,b_k)$, the condition
$\sum_{r=1}^{k}\left(E_{a_r}-E_{b_r}\right)=0$ implies
\begin{equation}
    (a_1,\ldots,a_k)=\pi(b_1,\ldots,b_k),
    \label{eq:k_nonresonance_condition}
\end{equation}
for some permutation $\pi\in S_k$.
\end{definition}

Under this condition, all surviving long-time contributions can be
organized by the permutations relating the two replica strings. The
corresponding time-averaged reduction is derived in \aref{sec:app_time_ave}. To exclude trivial intermediate operations,
such as the identity, whose trace remains of order $D$, we focus on
ensembles satisfying the following trace-suppression condition.

\begin{definition}[Trace-suppression condition]
\label{def:trace_suppression_condition}
An intermediate ensemble $\caE_{\mathrm{int}}$ is said to satisfy the
$k$-th order trace-suppression condition if
\begin{equation}
    \bbE_{U\sim\caE_{\mathrm{int}}}
    \left|\tr(U)\right|^{4k}
    =
    \caO_k(D^{2k-\epsilon}),
    \qquad
    \epsilon=\Omega_k(1),
    \label{eq:small_tr_ensemble}
\end{equation}
where $\caO_k(\cdot)$ and $\Omega_k(\cdot)$ denote asymptotic upper and lower
bounds on the scaling with $D$ for fixed $k$.
\end{definition}

This class includes many commonly used ensembles in quantum information
processing, such as the nontrivial Pauli set \(\caP_n^*\) and the Clifford
group \(\Cl_n\). At the same time, the condition excludes trivial choices such
as the identity operation $\{I^{\otimes n}\}$, as well as ensembles concentrated sufficiently close to it. Such perturbations are too weak to drive the evolution away from the
temporal orbit generated by the original fixed Hamiltonian.

\section{General criterion for unitary-design formation}

We now quantify how protocol $\caE_{\scrN}$ transfers
randomness from $\caE_{\mathrm{int}}$ to the generated
unitaries. For a GUE Hamiltonian, whose eigenbasis is Haar
distributed, we derive the frame potential of
$\caE_{\scrN}$ in terms of the lower-order frame potentials
of $\caE_{\mathrm{int}}$. The result extends to any
nonresonant spectrum with Haar-random eigenbases,
as stated in the following theorem;
see \aref{sec:app_prove_thm_main} for the proof.

\begin{theorem}
\label{thm:main_oneH_int}
Let the spectrum of \(H\) satisfy the additive nonresonance condition
\eref{eq:k_nonresonance_condition}, and let $\caE_{\mathrm{int}}$ satisfy
the trace-suppression condition \eref{eq:small_tr_ensemble}. Define the
Haar-eigenbasis-averaged frame potential by
\begin{equation}
    \overline{F}_{\caE_{\scrN}(H,\caE_{\mathrm{int}})}^{(k)}
    :=
    \bbE_{W\sim\haar}
    \left[
    F_{\caE_{\scrN}(W^{\dagger}HW,\caE_{\mathrm{int}})}^{(k)}
    \right].
\end{equation}
In the perfect time-filter limit \(T\to\infty\),
\begin{equation}
    \overline{F}_{\caE_{\scrN}(H,\caE_{\mathrm{int}})}^{(k)}
    =
    k!+\sum_{\ell=1}^{k}
    \frac{(k!)^2F_{\caE_{\mathrm{int}}}^{(\ell)}}
    {(k-\ell)!D^{2\ell}}
    +o_{D\to\infty}(1),
    \label{eq:oneH_main_result}
\end{equation}
where $o_{D\to\infty}(1)$ vanishes in the large-$D$ limit.
\end{theorem}

\Eref{eq:oneH_main_result} resolves the excess above the Haar value into
lower-order contributions from $\caE_{\mathrm{int}}$, each carrying a
factor $D^{-2\ell}$.  $\caE_{\scrN}$ approaches a $k$-design whenever the
frame potential of the intermediate ensemble $\caE_{\mathrm{int}}$ grows submaximally with $D$, as formalized by the following sufficient condition.

\begin{definition}[Frame-potential suppression criterion]
\label{def:intermediate_frame_potential_criterion}
An ensemble \(\caE\) satisfies the $k$-th order frame-potential
suppression criterion if
\begin{equation}\label{eq:small_frame_int}
    F_{\caE}^{(k)} =\caO_k( D^{2k-\epsilon}),
    \qquad
    \epsilon=\Omega_k(1).
\end{equation}
\end{definition}
Frame-potential suppression at any one fixed order implies suppression
at every fixed order (\pref{prp:lower_trace_suppression} in \sref{sec:app_prove_thm_main}).
Thus, under the assumptions of \thref{thm:main_oneH_int}, it suffices
to verify \eref{eq:small_frame_int} at first order for every correction
in \eref{eq:oneH_main_result} to vanish. A common choice for the intermediate ensemble is the uniform distribution over
a unitary group \(\rmG\), including, for example, the Pauli and Clifford
groups. Because \(\rmG\) is closed under multiplication and its uniform
measure is invariant under group translations, the relative element
\(U^\dagger V\) is itself uniformly distributed over \(\rmG\). Consequently,
\begin{equation}
    F_{\rmG}^{(\ell)}
    =
    \bbE_{U,V\sim\rmG}
    \left|\tr(U^\dagger V)\right|^{2\ell}
    =
    \bbE_{V\sim\rmG}
    \left|\tr(V)\right|^{2\ell}.
\end{equation}
Hence, for a uniformly sampled group, the trace-suppression condition
in \eref{eq:small_tr_ensemble} bounds $F_{\rmG}^{(2k)}$ and implies
frame-potential suppression at every fixed order by
\pref{prp:lower_trace_suppression}.

Although \thref{thm:main_oneH_int} is formulated after averaging over the
eigenbasis of the Hamiltonian, this averaged result already implies
typicality whenever $\caE_{\mathrm{int}}$ satisfies both the
trace-suppression condition in \eref{eq:small_tr_ensemble} and the
frame-potential suppression criterion in \eref{eq:small_frame_int}. The key
observation is that the Haar value $k!$ is a universal lower bound on the
$k$-th frame potential. Hence, for each fixed eigenbasis $W$, we may define
the nonnegative deviation
\begin{equation}
    \Delta_W^{(k)}
    :=
    F_{\caE_{\scrN}(W^\dagger H W,\caE_{\mathrm{int}})}^{(k)}-k!
    \geq 0.
\end{equation}
Under the above conditions, \eref{eq:oneH_main_result} gives
$\bbE_W\Delta_W^{(k)}=o_{D\to\infty}(1)$. Since this deviation is
nonnegative, its vanishing average cannot result from cancellations between
different eigenbases. Indeed, for any fixed $\eta>0$, Markov's inequality
yields
\begin{equation}
    \Pr_W\!\left[\Delta_W^{(k)}>\eta\right]
    \leq
    \frac{\bbE_W\Delta_W^{(k)}}{\eta}
    =
    o_{D\to\infty}(1).
\end{equation}
Thus, for fixed \(k\),
\(F_{\caE_{\scrN}}^{(k)}=k!+o_{D\to\infty}(1)\) with high probability over
Haar-random eigenbases. The eigenbasis average in
\thref{thm:main_oneH_int} therefore establishes a typicality result rather
than a conclusion restricted to the average alone. Together with the numerical simulations based on the Rydberg Hamiltonian below, this supports the applicability of the mechanism beyond random-matrix models in a structured, experimentally relevant chaotic system. Although
\thref{thm:main_oneH_int} assumes a perfect time filter, the finite-$T$
numerical results in \fref{fig:finite_time_deltaF}, consistent with
\rscite{Zhou2026three,sun2026unitary}, reach Haar-level frame potentials at moderate sampling windows for the system sizes studied.

\section{Design-forming and design-limiting intermediate ensembles}
\subsection{Design-forming ensembles}
We now identify standard
choices of $\caE_{\mathrm{int}}$ that satisfy the above sufficient
criterion and drive $\caE_{\scrN}$ toward a unitary $k$-design in the
perfect time-filter and large-$D$ limits.

\begin{figure}[!tbp]
    \centering
    \includegraphics[width=\columnwidth]{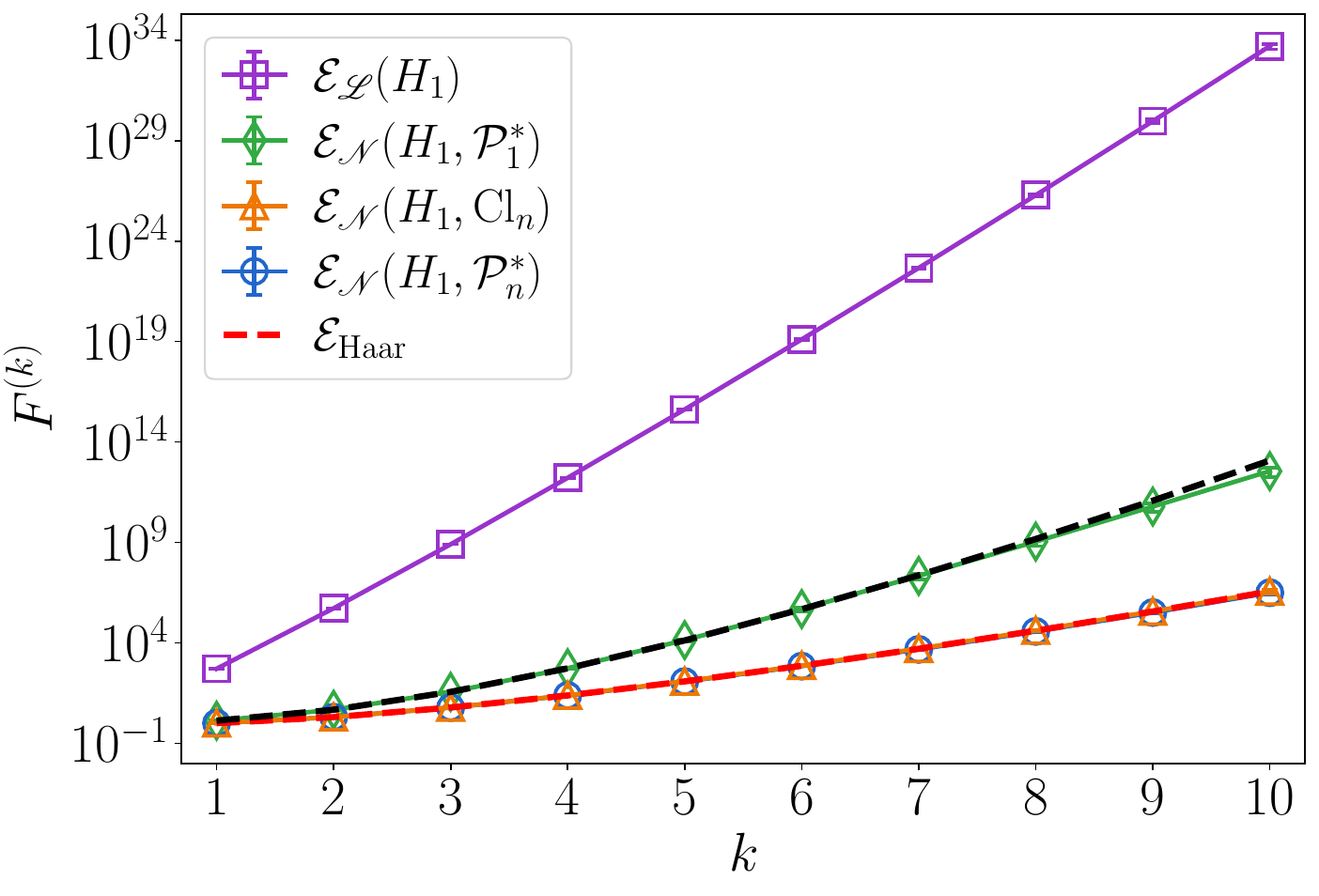}
\caption{Frame potential $F^{(k)}$ versus $k$ for the relevant ensembles.
Here $H_1$ is a single $n=9$-qubit GUE Hamiltonian, fixed throughout each
protocol.  Evolution times are sampled uniformly from $[0,T]$, with
$T=10^{8}$ for $\caE_{\scrL}(H_1)$ and
$\caE_{\scrN}(H_1,\caP_1^*)$, and $T=10^3$ for
$\caE_{\scrN}(H_1,\Cl_n)$ and $\caE_{\scrN}(H_1,\caP_n^*)$.
Each point averages $10^7$ samples.  The black dashed lines are the
analytical prediction for $\caE_{\scrN}(H_1,\caP_1^*)$ from
\eqsref{eq:oneH_main_result}{eq:local_pauli_frame}; the full Pauli and
Clifford protocols coincide with the Haar values, and all the error bars are smaller than the markers in the plot.}
\label{fig:Ensembles_demonstration}
\end{figure}

The full nontrivial Pauli set $\caP_n^*$ is a simple design-forming
intermediate ensemble.  Every element is traceless, so
\eref{eq:small_tr_ensemble} holds, while Hilbert--Schmidt orthogonality
gives
\begin{equation}
    F_{\caP_n^*}^{(\ell)}
    =
    \frac{D^{2\ell}}{|\caP_n^*|}
    =
    \frac{D^{2\ell}}{D^2-1}
    \approx D^{2\ell-2}
    \ll D^{2\ell}.
    \label{eq:global_Pauli_frame}
\end{equation}
Thus $F_{\caP_n^*}^{(\ell)}$ is suppressed by $D^{-2}$ relative to its
maximal scaling, and $\caE_{\scrN}(H,\caP_n^*)$ approaches the Haar frame
potential, as shown in \fref{fig:Ensembles_demonstration}.  This protocol
is structurally similar to $\caE_{\scrM}$ in \rcite{sun2026unitary}, but
uses the same Hamiltonian setting before and after the Pauli pulse rather
than two independently controlled settings.  The reduced control
requirement is particularly advantageous for analog many-body platforms.

The $n$-qubit Clifford group $\Cl_n$ gives a second example.  Although it
is a unitary $3$-design but not a unitary $4$-design
\cite{Webb2016,zhu2016clifford,Zhu2017}, its higher-order frame potentials
are known exactly.  For $n\ge \ell-1$ \cite{Gross2021},
\begin{equation}
    F_{\Cl_n}^{(\ell)}
    =
    \prod_{j=0}^{\ell-2}(2^j+1).
    \label{eq:clifford_frame}
\end{equation}
The $n$-qubit Clifford group $\Cl_n$ gives a second example. Although it
is a unitary $3$-design but not a unitary $4$-design
\cite{Webb2016,zhu2016clifford,Zhu2017}, its higher-order frame potentials
are known exactly. For $n\ge \ell-1$ \cite{Gross2021},
\begin{equation}
    F_{\Cl_n}^{(\ell)}
    =
    \prod_{j=0}^{\ell-2}(2^j+1).
    \label{eq:clifford_frame}
\end{equation}
This expression grows with $\ell$ but is independent of $D=2^n$.
Consequently, $F_{\Cl_n}^{(k)}=\caO_k(1)$ for fixed $k$, satisfying
\eref{eq:small_frame_int}. The group-translation identity gives
$\bbE_{U\sim\Cl_n}|\tr(U)|^{4k}=F_{\Cl_n}^{(2k)}=\caO_k(1)$,
which also guarantees \eref{eq:small_tr_ensemble}.
Thus, inserting a perturbation drawn from the Pauli or Clifford
ensembles between two independently timed evolutions under the same
chaotic Hamiltonian can substantially enhance their randomness,
yielding Haar-level frame potentials beyond their intrinsic design
orders. This enhancement is illustrated numerically in
\fref{fig:Ensembles_demonstration}.

\begin{figure*}[!tp]
    \centering
    \includegraphics[width=0.95\textwidth]{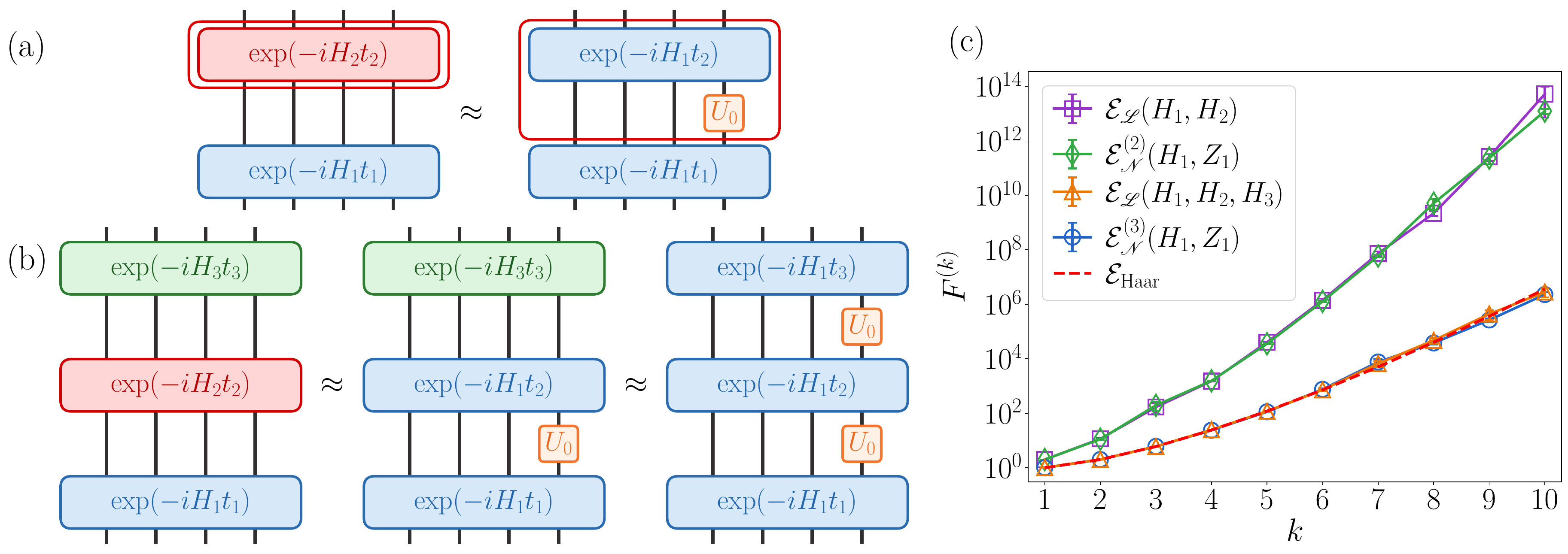}
\caption{Repeated fixed perturbations as a control-efficient substitute
for independent Hamiltonian evolutions. The evolution times
$t_1,t_2,t_3$ are sampled uniformly from $[0,T]$ in the long-time limit
$T\to\infty$, and $U_0$ denotes a fixed trace-suppressed
unitary perturbation acting on the first qubit. (a) A single perturbation produces
the same  frame potential as the two-Hamiltonian protocol, as in
\eref{eq:recursive_2sp}. (b) With two perturbation insertions, the one-Hamiltonian protocol produces
the same  frame potential as the three-Hamiltonian protocol, as established in \thref{thm:three_round_fixed_perturbation}. (c) Frame potential $F^{(k)}$
as a function of $k$ for the relevant ensembles. The chaotic
Hamiltonians $H_1,H_2,H_3$ act on $n=9$ qubits and, after being sampled
from the GUE, are fixed throughout the protocol. The evolution times
are sampled uniformly from $[0,T]$, with $T=10^{8}$ for
$\caE_{\scrL}(H_1,H_2)$ and
$\caE^{(2)}_{\scrN}(H_1,Z_1)$, and $T=10^6$ for
$\caE_{\scrL}(H_1,H_2,H_3)$ and
$\caE^{(3)}_{\scrN}(H_1,Z_1)$. Each point is averaged over $10^7$
random samples, and all the error bars are smaller than the markers in the plot.}
\label{fig:equivalent_construction}
\end{figure*}

\subsection{Design-limiting ensembles}
We next consider two
distinct obstructions to design formation: the identity layer violates
trace suppression, whereas a fixed-support Pauli ensemble is traceless
but still fails frame-potential suppression.

For the identity layer $\caE_{\mathrm{int}}=\{I^{\otimes n}\}$,
$\caE_{\scrN}$ reduces to the purely temporal ensemble $\caE_{\scrL}$ of
a single Hamiltonian.  Since $\tr(I^{\otimes n})=D$, trace suppression
fails and \thref{thm:main_oneH_int} does not apply.  Instead, as proved in
\aref{sm:SMFrameLH},
\begin{align}
    F^{(k)}_{\caE_{\scrL}(H)}
    &=
    \sum_{\substack{\gamma_1+\cdots+\gamma_D=k\\ \gamma_i\ge0}}
    \left(
        \frac{k!}{\gamma_1!\cdots \gamma_D!}
    \right)^2\notag\\
    &=
    k!D^k+\caO_k(D^{k-1}),
    \label{eq:identity_intermediate_frame_potential}
\end{align}
reproducing the temporal-ensemble result of \rcite{Roberts2017chaos}.
The leading term exceeds the Haar value $k!$ by a factor $D^k$, as shown
in \fref{fig:Ensembles_demonstration}: without perturbations, a fixed
Hamiltonian explores only phase rotations in its energy basis instead of the unitary group.

A subtler example is the local nontrivial Pauli ensemble $\caP_w^*$
($w\le n$), which contains Pauli strings on a $w$-qubit
subsystem and acts trivially elsewhere.  Every element is traceless, so
\eref{eq:small_tr_ensemble} holds, but
\begin{equation}
    F^{(\ell)}_{\caP_w^*}
    =
    \frac{D^{2\ell}}{|\caP_w^*|}
    =
    \frac{D^{2\ell}}{2^{2w}-1}.
    \label{eq:local_pauli_frame}
\end{equation}
For fixed $\ell\le k$ and fixed $w$ independent of $n$, this remains
$\Theta_w(D^{2\ell})$ and therefore violates \eref{eq:small_frame_int}.
The $w=1$, $n=9$ data in \fref{fig:Ensembles_demonstration} show this
failure directly.  

More generally, the diagonal terms $U=V$ imply that
any finite unitary set $\rmS$ obeys
\begin{equation}
    F_{\rmS}^{(\ell)}
    =
    \frac{1}{|\rmS|^2}
    \sum_{U,V\in\rmS}
    |\tr(U^\dagger V)|^{2\ell}
    \ge
    \frac{D^{2\ell}}{|\rmS|}.
    \label{eq:finite_ensemble_lower_bound}
\end{equation}
Consequently, satisfying \eref{eq:small_frame_int} requires at least
$|\rmS|=\Omega_k(D^\epsilon)$ for some $\epsilon=\Omega_k(1)$; an
ensemble of subpolynomial size retains an order-$D^{2\ell}$ frame
potential up to subpolynomial factors.  The identity layer and fixed-$w$
Pauli ensemble realize these two design-limiting mechanisms.

\section{Replacing independent Hamiltonian controls with repeated fixed perturbations}

Beyond its direct application to the one-Hamiltonian protocol,
\thref{thm:main_oneH_int} provides a useful setting for comparing the
two ensemble constructions. Consider the extremal intermediate
ensemble $\widetilde{\caE}=\{U_0\}$, where $U_0$ is a fixed trace-suppressed 
unitary. This corresponds to ensembles with the maximal frame potential
$F_{\widetilde{\caE}}^{(\ell)}=D^{2\ell}$ and therefore violates
\eref{eq:small_frame_int}.
Nevertheless, when inserted between two temporal evolutions generated
by a chaotic Hamiltonian, \thref{thm:main_oneH_int} gives
\begin{align}
\overline{F}_{\caE_{\scrN}(H,\widetilde{\caE})}^{(k)}
&=k!+\sum_{\ell=1}^{k}\frac{(k!)^2}{(k-\ell)!}
+o_{D\to\infty}(1).
\label{eq:fixed_traceless_frame}
\end{align}

Thus, a single fixed perturbation remains insufficient to reach the Haar
value, but it already reduces the purely temporal result from order
$D^k$ to $\caO_k(1)$. This strong reduction shows that the perturbation redirects the subsequent chaotic evolution and thereby unlocks randomness already present in the temporal dynamics. Surprisingly, the leading-order frame potential in
\eref{eq:fixed_traceless_frame} coincides with that of the two-Hamiltonian
temporal protocol $\caE_{\scrL}(H_1,H_2)$ in
\rcite{Zhou2026three}. Hence, at the level of the $k$-th frame potential,
the effect of a second independent chaotic Hamiltonian can be reproduced
by inserting a fixed trace-suppressed perturbation:
\begin{equation}
    F_{\caE_{\scrL}(H_1,H_2)}
    \approx
    F_{\caE_{\scrN}(H_1,\{U_0\})}.
    \label{eq:recursive_2sp}
\end{equation}

This observation motivates an elementary replacement rule for
Hamiltonian-based protocols. We extend this correspondence to multiple
rounds of fixed perturbations by defining
\begin{equation}
\caE^{(m)}_{\scrN}(H_1,U_0)
:=
\left\{
\prod_{i=1}^{m-1}\left(
\rme^{-\rmi H_1t_i}U_0
\right)
\rme^{-\rmi H_1t_m}
\right\},
\label{eq:recursive_oneH_protocol}
\end{equation}
with $(m-1)$ rounds of perturbations and
$\caE^{(2)}_{\scrN}(H_1,U_0)=\caE_{\scrN}(H_1,\{U_0\})$.
As proved in  \aref{sec:app_frame_potential_monotonicity}, the frame potential is monotonically nonincreasing with $m$:
\begin{equation}
F_{\caE^{(m+1)}_{\scrN}(H_1,U_0)}^{(k)}
\leq
F_{\caE^{(m)}_{\scrN}(H_1,U_0)}^{(k)}.
\label{eq:recursive_monotonicity}
\end{equation}
The following theorem provides a rigorous analytical guarantee for
replacing multiple independent Hamiltonians with repeated applications
of a fixed trace-suppressed perturbation while still realizing unitary
$k$-designs. The proof is provided in  \aref{sec:proofthm2}.

\begin{theorem}
\label{thm:three_round_fixed_perturbation}
Let the spectrum of $H_1$ satisfy the additive nonresonance condition
\eref{eq:k_nonresonance_condition}, and let $U_0$ be a fixed trace-suppressed
unitary. In the perfect time-filter limit, for $m\ge3$,
\begin{equation}
\overline{F}_{\caE^{(m)}_{\scrN}(H_1,U_0)}^{(k)}
=k!+o_{D\to\infty}(1).
\label{eq:three_round_fixed_perturbation}
\end{equation}
\end{theorem}
By the same argument as above, the result holds for
Haar-typical eigenbases, although exceptional commuting choices
$[H_1,U_0]=0$ reduce the protocol to a purely temporal orbit. Hence, two perturbations reproduce the asymptotic unitary $k$-design
property of the three-Hamiltonian protocol
$\caE_{\scrL}(H_1,H_2,H_3)$ proposed in
\rcite{Zhou2026three}, while additional rounds preserve it. Independent
Hamiltonian controls can therefore be replaced by one fixed chaotic
Hamiltonian, random evolution times, and a repeatable pulse such as a
local Pauli operation, substantially reducing control overhead and
experimental resource requirements while improving practical
feasibility.

\section{Physically accessible models}

\subsection{Disordered Rydberg dynamics}
Most analytical and numerical
results above are obtained for GUE Hamiltonians, which provide a clean model
of generic chaotic dynamics but are difficult to implement directly in
current experimental platforms. While similar Hamiltonian-based protocols
have also been studied with SYK-type chaotic models
\cite{Zhou2026realizing,Zhou2026three,Pikulin2017black,Asaduzzaman2024syk,
Cao2020Towards}, here we focus on a Rydberg-atom model that is more directly
relevant to near-term platforms
\cite{Labuhn2016Rydberg,Marcuzzi2017facilitation,Oliveira2025demo}.
We consider the one-dimensional Rydberg Hamiltonian with open boundary
conditions:
\begin{equation}\label{eq:rydHamiltonian}
    H_{\mathrm{Ryd}}
    =
    \sum_i \frac{\Omega_i}{2}X_i
    -
    \sum_i \Delta_i \hat n_i
    +
    \sum_{i<j} V_{ij} \hat n_i\hat n_j ,
\end{equation}
where $\hat n_i$ is the Rydberg occupation operator, and
$V_{ij}=V_0/|i-j|^6$ describes the van der Waals interaction between
Rydberg excitations. We take
$\Omega_i=\Omega_0(1+\epsilon_i)$ and
$\Delta_i=\Delta_0(1+\eta_i)$, where the weak random inhomogeneities
$\epsilon_i$ and $\eta_i$ are introduced to break spatial symmetries and
remove accidental degeneracies and resonances. Driven Rydberg arrays
naturally realize this model through laser-induced Rabi drives, tunable
detunings, and van der Waals density-density interactions, with weak
inhomogeneities accessible through site-dependent controls and atom
positions.

As shown in \fref{fig:Rydberg_ensembles_comparison}, we numerically study an
$n=9$-qubit system governed by the Hamiltonian in
\eref{eq:rydHamiltonian}. We set $\Omega_0=\Delta_0=1$ and $V_0=1.5$,
and sample the inhomogeneities independently and uniformly from
$\epsilon_i\in[-0.12,0.12]$ and
$\eta_i\in[-0.35,0.35]$. The spectral statistics shown in the lower panel of
\fref{fig:Rydberg_ensembles_comparison} provide direct evidence that
the Rydberg model lies in the quantum-chaotic regime. The level-spacing distribution is well described by the
Gaussian orthogonal ensemble (GOE) and clearly deviates from Poisson
statistics, exhibiting the characteristic level repulsion of chaotic
spectra. The GOE is the appropriate random-matrix
benchmark here because $H_{\mathrm{Ryd}}$ is real and therefore
preserves an antiunitary time-reversal symmetry.

\begin{figure}[!tbp]
    \centering
    \includegraphics[width=0.9\columnwidth]{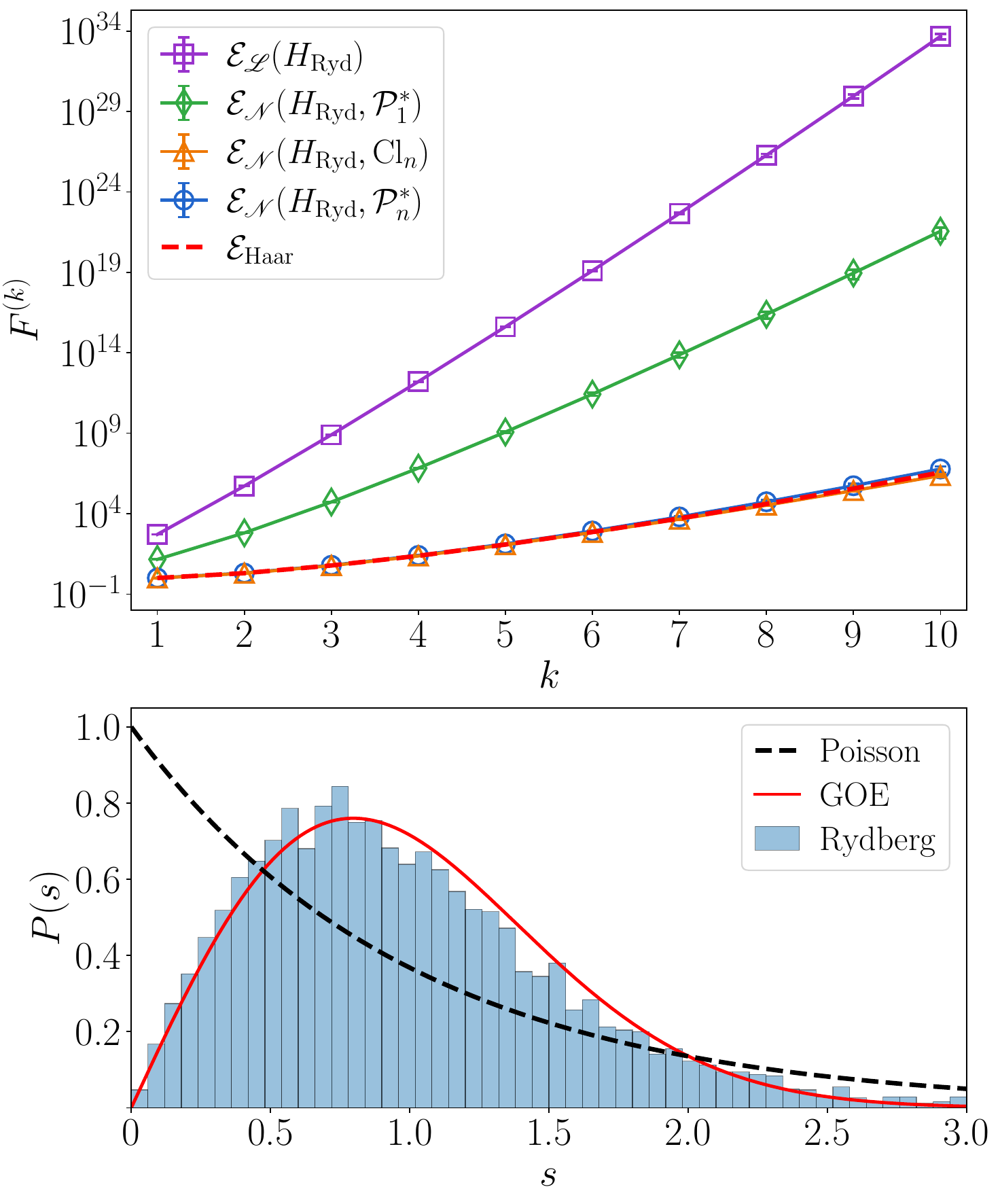}
    \caption{
    Frame-potential and spectral statistics of the disordered Rydberg
    Hamiltonian in \eref{eq:rydHamiltonian}. In the upper panel,
    $F^{(k)}$ is shown as a function of $k$ for the relevant ensembles.
    Here $H_{\mathrm{Ryd}}$ acts on $n=9$ qubits, with
    $\Omega_i=\Omega_0(1+\epsilon_i)$ and
    $\Delta_i=\Delta_0(1+\eta_i)$, where $\epsilon_i$ and $\eta_i$ are
    sampled independently and uniformly from $[-0.12,0.12]$ and
    $[-0.35,0.35]$, respectively. The evolution times are sampled uniformly
    from $[0,T]$ with $T=10^{10}$, approximating the perfect time-filter
    limit. For each realization of $H_{\mathrm{Ryd}}$, every frame-potential
    data point is averaged over $10^7$ random samples, and the final results
    are further averaged over 20 independent Hamiltonian realizations.
    The data for $\caE_{\scrN}(H_{\mathrm{Ryd}},\caP_n^*)$ and
    $\caE_{\scrN}(H_{\mathrm{Ryd}},\Cl_n)$ coincide with the Haar values, and all the error bars are smaller than the markers in the plot.
    The lower panel shows the nearest-neighbor level-spacing distribution
    $P(s)$ for the same Rydberg model. The solid and dashed curves denote
    the GOE Wigner--Dyson and Poisson predictions, respectively. The
    agreement with GOE statistics and the associated level repulsion provide
    spectral evidence that the chosen Rydberg parameters lie in the
    quantum-chaotic regime.
    }
    \label{fig:Rydberg_ensembles_comparison}
\end{figure}

Consistent with this chaotic spectral behavior, the full nontrivial Pauli
ensemble $\caP_n^*$ and the Clifford ensemble $\Cl_n$ act as effective
design-forming intermediate ensembles. When inserted into
$\caE_{\scrN}(H_{\mathrm{Ryd}},\caE_{\mathrm{int}})$, they drive the
resulting one-Hamiltonian protocol to Haar-level frame potentials, as
shown in the upper panel of \fref{fig:Rydberg_ensembles_comparison}.
This behavior is consistent with their trace- and frame-potential
suppression properties. By contrast, the local Pauli ensemble provides
additional randomness relative to the unperturbed temporal ensemble, but
remains far from sufficient to generate a unitary $k$-design.

\section{Universal space--depth lower bound for local models}
The effectiveness of repeated perturbations depends on their relation to the
Hamiltonian. For local dynamics, a perturbation acting on only a few sites
may leave much of the Hamiltonian unchanged, even after several repetitions.
This constraint is obscured by Haar-random eigenbases: in the GUE setting,
nontrivial Pauli strings with different supports are unitarily equivalent
and therefore have identical eigenbasis-averaged statistics. For physical
local Hamiltonians, however, the spatial support of the control is a
resource that must be considered together with the number of perturbations.

To quantify the effect of this spatial constraint, we consider
finite-range geometrically local Hamiltonians with local interaction
strengths that remain bounded as the system size increases. Since an
additive identity component contributes only a global phase and leaves
all frame potentials unchanged, we choose the energy origin such that
$\tr(H)=0$ throughout the resource analysis. The Hamiltonian can then
be expanded as
$H=\sum_{\alpha\ne 0}h_\alpha P_\alpha$.

For a standard geometrically local spin Hamiltonian, each site
participates in only $O(1)$ finite-range interaction terms of $O(1)$
strength. Consequently, the total squared interaction weight associated
with any site remains bounded in the thermodynamic limit. We therefore
assume
\begin{align}
    \sup_{1\le i\le n}
    \sum_{\alpha:\,i\in\operatorname{supp}(P_\alpha)}
    |h_\alpha|^2
    &\le J_{\mathrm{loc}}^2,
    \label{eq:local_squared_weight_bound}\\
    \sigma_H^2:=\frac{\tr(H^2)}{D}
    =\sum_{\alpha\ne 0}|h_\alpha|^2
    &\ge vn,
    \label{eq:extensive_energy_variance_bound}
\end{align}
where $\operatorname{supp}(P_\alpha)$ denotes the nonidentity sites of
$P_\alpha$, and the constants $J_{\mathrm{loc}}$ and $v$ are independent
of $n$. The first condition is thus a direct consequence of the usual
bounded-strength locality assumption, while the second requires a
nonvanishing density of local Hamiltonian terms. Moreover,
\eref{eq:local_squared_weight_bound} also implies
$\sigma_H^2\le nJ_{\mathrm{loc}}^2$. Together,
\eqsref{eq:local_squared_weight_bound}{eq:extensive_energy_variance_bound}
therefore give
$\sigma_H^2=\Theta(n)$,
so that the infinite-temperature energy variance is extensive.
These assumptions are satisfied by a broad class of bounded-strength
finite-range spin Hamiltonians, including transverse- and mixed-field
Ising models, Heisenberg spin models, and Kitaev-type spin
models~\cite{Bethe1931,Pfeuty1970Ising,Kitaev2006,KimHuse2013}.

Let $\mathcal E_{\mathrm{int}}$ be an ensemble of Hermitian unitary
perturbations. Define its average support by
\begin{equation}
    \bar w:=\mathbb E_{U\sim\mathcal E_{\mathrm{int}}}
    \bigl|\operatorname{supp}(U)\bigr|,
    \label{eq:mean_control_support}
\end{equation}
which measures the spatial extent
of the required control. We consider $(m+1)$
evolution intervals separated by $m$ perturbations,
\begin{equation}
    V_{\boldsymbol t,\boldsymbol U}
    =e^{-iHt_{m+1}}U_m e^{-iHt_m}\cdots U_1e^{-iHt_1},
    \label{eq:local_repeated_control_protocol}
\end{equation}
where the $U_j$ are sampled independently from the same ensemble
$\caE_{\mathrm{int}}$. The repeated fixed Hermitian perturbation discussed
before is included by taking $\mathcal E_{\mathrm{int}}$ to contain a
single element. Write $F_{m+1,\infty}^{(k)}$ for the frame potential of
this ensemble in the limit $T\to\infty$, and denote its relative excess
above the Haar value by
$\delta F_{m+1,\infty}^{(k)}
:=(F_{m+1,\infty}^{(k)}-k!)/k!$.
We define the expected control area
\begin{equation}
    \mathcal A_{\mathrm{ctrl}}
    :=\sum_j\bbE_{U_j\sim\caE_{\mathrm{int}}}
    \bigl|\operatorname{supp}(U_j)\bigr|
    =m\bar w,
    \label{eq:control_area_definition}
\end{equation}
as shown in \tref{tab:control_area}.

\begin{table}[t]
    \centering
    \caption{
    Numerical illustration of the control area $\mathcal A_{\rm ctrl}$
    defined in \eref{eq:control_area_definition}.
    Results are obtained for the Rydberg model in
    \eref{eq:rydHamiltonian} with $n=9$ qubits and
    $\caE_{\mathrm{int}}(\omega)=\{X,Y,Z\}^{\otimes\omega}$.
    Here $m_{\rm th}$ is the theoretical lower bound from
    \eref{eq:local_control_frame_bound1} in
    \thref{thm:local_space_depth_bound}, while $m_*$ is the
    minimum depth found numerically to achieve
    $\delta F^{(3)}\le 0.1$.
    Evolution times are sampled uniformly from $[0,T]$
    with $T=10^{10}$, approximating the perfect time-filter
    limit. For each realization of $H_{\mathrm{Ryd}}$,
    each data point is averaged over $10^7$ random unitary
    samples, followed by an average over 20 independent
    Hamiltonian realizations.
    The resulting control areas are identical.
    }
    \label{tab:control_area}
    \vspace{0.5ex}
    \renewcommand{\arraystretch}{0.85}
    \begin{tabular*}{\columnwidth}{@{\extracolsep{\fill}}cccc}
        \toprule
        $\omega$ & $m_{\rm th}$ & $m_*$ & $\mathcal A_{\rm ctrl}$ \\
        \midrule
        1 & 12 & 12 & 12 \\
        2 & 6  & 6  & 12 \\
        3 & 4  & 4  & 12 \\
        4 & 3  & 3  & 12 \\
        6 & 2  & 2  & 12 \\
        \bottomrule
    \end{tabular*}
\end{table}

\begin{theorem}[Lower bound on the space--depth control]
\label{thm:local_space_depth_bound}
Suppose that $H$ satisfies
\eqsref{eq:local_squared_weight_bound}{eq:extensive_energy_variance_bound}
and that $\mathcal E_{\rm int}$ is an ensemble of Hermitian unitaries.
For integers $k,m\ge 1$ with $D\ge k+1$, we have
\begin{align}
    \delta F_{m+1,\infty}^{(k)}
    &\ge \frac{k}{D^{2m}\sigma_H^{4m}}\left|
    \mathbb E_{U\sim\mathcal E_{\rm int}}
    \operatorname{Tr}(HUHU^\dagger)
    \right|^{2m} \label{eq:local_control_frame_bound1}\\
    &\ge k\left(
    1-\frac{2J_{\mathrm{loc}}^2\bar w}{vn}
    \right)_+^{2m},\label{eq:local_control_frame_bound2}
\end{align}
where $(x)_+:=\max\{x,0\}$. If $J_{\mathrm{loc}}^2\bar w/(vn)\le 1/4$,
achieving $\delta F_{m+1,\infty}^{(k)}\le\varepsilon$ necessarily requires
\begin{equation}
    \mathcal A_{\mathrm{ctrl}}=m\bar w
    \ge\frac{vn}{8J_{\mathrm{loc}}^2}\,
    \log\left(\frac{k}{\varepsilon}\right).
    \label{eq:universal_space_depth_bound}
\end{equation}
\end{theorem}

The result is universal within the stated assumptions, requiring neither
an additive nonresonance condition nor averaging over Haar-random
eigenbases. Note, however, that it provides only a necessary condition
on the control complexity; satisfying the bound does not by itself
guarantee convergence to the Haar frame potential.

The bound follows from the limited fraction of terms in the Hamiltonian
that can be affected in each control round. A perturbation supported on
a set $S$ commutes with every Pauli component of $H$ whose support is
disjoint from $S$. By \eref{eq:local_squared_weight_bound}, the total
squared weight of the components that can be modified is at most
$J_{\mathrm{loc}}^2|S|$, whereas
\eref{eq:extensive_energy_variance_bound} ensures that the total squared
weight of $H$ is at least $vn$. Averaging over the perturbation ensemble
and applying this estimate successively across the $m$ control rounds
yields the factor
$\left[1-2J_{\mathrm{loc}}^2\bar w/(vn)\right]^{2m}$ in
\eref{eq:local_control_frame_bound2}.

This universal constraint also admits an interesting interpretation
in terms of energy-space localization~\cite{YuWei2026EnergySpace}.
For local Hamiltonians with bounded local strength $ \sup_{1\le i\le n}\sum_{\alpha:\,i\in\operatorname{supp}(P_\alpha)}
    |h_\alpha|=L=O(1)$,
the cumulative Hamiltonian variation
$\Lambda:=\int\|\mathrm dH(t)\|_X$, where $\|\cdot\|_X$
denotes the sum of the norms of the local terms
($\|H\|_X=\sum_{\alpha\ne0}|h_\alpha|$ for the
Pauli decomposition above), controls the characteristic energy-space
localization window $[E_0-\Lambda,E_0+\Lambda]$
of an initial energy eigenstate with energy $E_0$,
with exponentially suppressed leakage into extensively
more distant energy regions.

The inserted unitary perturbations can be viewed as time-dependent
controls or changes of the Hamiltonian in a rotating
frame. A perturbation supported on $S_j$ affects only those terms
 whose support intersects $S_j$. Each such term
changes by at most $2|h_\alpha|$ in operator norm. Summing these
changes over the original local terms and all control rounds gives
$\Lambda\leq 2L\sum_j |S_j|$ for each realization. Averaging over
the perturbations yields
\begin{equation}
    \mathbb E[\Lambda]\leq 2L\mathcal A_{\mathrm{ctrl}}.
    \label{eq:control_area_energy_variation}
\end{equation}
Insufficient control area therefore restricts transport between
extensively separated energy regions, regardless of how long the
system evolves between perturbations. While this does not preclude
thermalization within an energy shell, full-Hilbert-space
unitary-design formation requires redistributing the initially sharp energy distribution of a spectral-edge eigenstate toward the infinite-temperature energy distribution. An extensive control area,
$\mathcal A_{\mathrm{ctrl}}=\Omega(n)$, is thus necessary,
consistent with the linear system-size dependence in
\eref{eq:universal_space_depth_bound}. Our frame-potential bound
further quantifies this resource requirement in terms of the
design order and target accuracy.

For the numerical investigation, we consider the $n=9$ Rydberg Hamiltonian
introduced above and determine, for each perturbation support
$\omega$, the minimum number of perturbation rounds $m_*$ required to
reach $\delta F^{(3)}\le0.1$. As shown in \tref{tab:control_area}, the corresponding theoretical estimate
$m_{\rm th}$ decreases
approximately inversely with $\omega$, as expected from the
space--depth tradeoff. The numerical scan exhibits a similar tradeoff: increasing the perturbation support reduces the required depth, while the resulting control area $m^*\omega$ remains approximately constant
over the range studied, consistent with the resource dependence
of the lower bound in \thref{thm:local_space_depth_bound}.Thus, $\mathcal A_{\rm ctrl}$ also
emerges empirically as an effective resource scale for the local dynamics.

\begin{figure*}[!tp]
    \centering
    \includegraphics[width=0.9\textwidth]
    {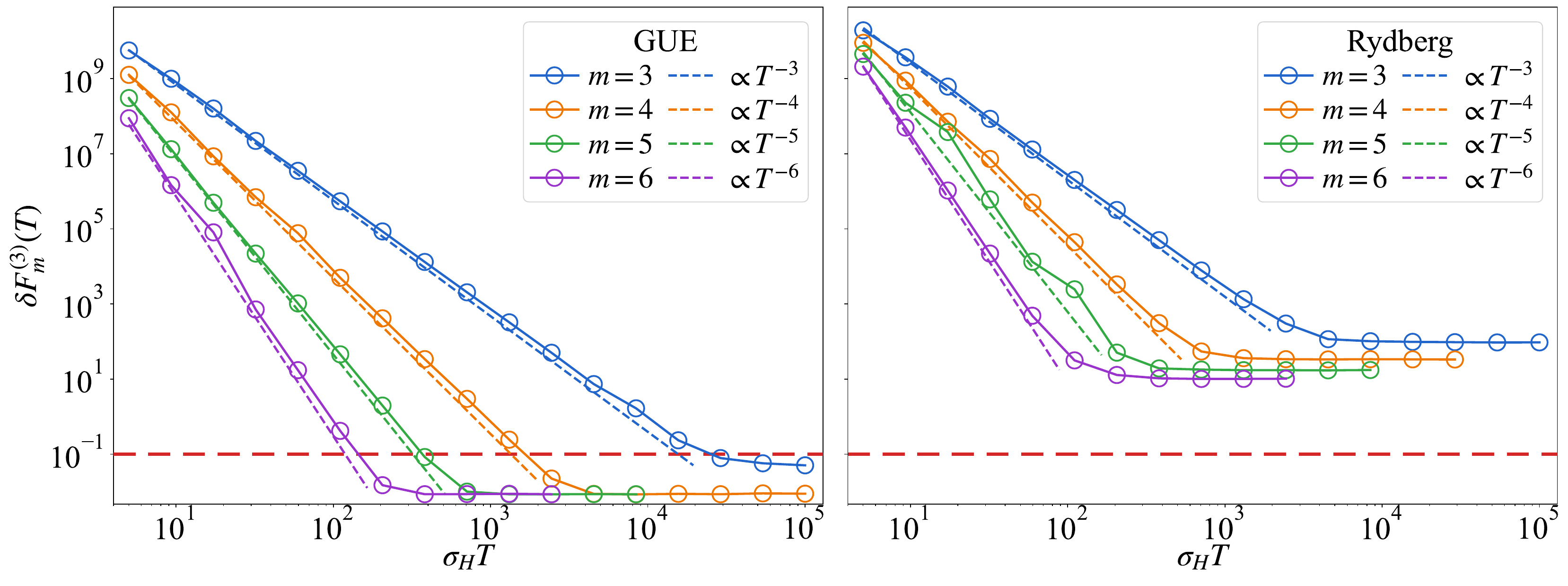}
    \caption{
    The normalized deviation of the frame potential
    $\delta F_{m}^{(3)}(T)$ as a function of $\sigma_H T$ for the relevant
    ensembles. The left and right panels show the results for GUE (left)
    and Rydberg (right) dynamics for $n=7$ qubits, respectively.
    The perturbation is a fixed $\sigma_z$ operator acting on the first
    qubit, with $m$ denoting the number of evolution intervals separated
    by $(m-1)$ perturbations. The left panel shows the GUE Hamiltonian,
    which is fixed throughout the protocol after being sampled.
    The right panel corresponds to the Rydberg Hamiltonian
    $H_{\mathrm{Ryd}}$ in \eref{eq:rydHamiltonian}, with
    $\Omega_i=\Omega_0(1+\epsilon_i)$ and
    $\Delta_i=\Delta_0(1+\eta_i)$. The inhomogeneities $\epsilon_i$
    and $\eta_i$ are sampled independently and uniformly from
    $[-0.12,0.12]$ and $[-0.35,0.35]$, respectively.
    For the GUE Hamiltonian, each point is averaged over $10^7$
    independent samples. For each realization of $H_{\mathrm{Ryd}}$,
    every data point is averaged over $10^7$ random samples of the
    corresponding unitary ensemble; the final data points are further
    averaged over 20 independent realizations of $H_{\mathrm{Ryd}}$.
    The dashed lines indicate the corresponding reference scalings
    proportional to $(\sigma_H T)^{-m}$ for different values of $m$.
    The horizontal red dashed line marks the reference threshold
    $\delta F=0.1$. The nonvanishing long-time plateau for the Rydberg
    dynamics is consistent with the insufficient spatial support of
    the single-qubit perturbation, as captured by the space--depth
    constraint in \eref{eq:universal_space_depth_bound}.
    }
    \label{fig:finite_window_scaling}
\end{figure*}

\Eref{eq:universal_space_depth_bound} therefore identifies a
spatial resource requirement that complements the results based on
random matrix models. In particular, realizing those constructions
with only a few fixed perturbations requires sufficiently extended
perturbations when the Hamiltonian satisfies the locality assumptions,
or sufficiently nonlocal Hamiltonians that fall outside the stated
locality assumptions. Their implementation thus depends on the spatial
structure of both the perturbations and the Hamiltonian, which is not
resolved by a description based on Haar-random eigenbases.
The statistical properties of the intermediate ensemble remain essential
to convergence; the present bound only isolates an additional constraint
imposed by locality.

\section{Universal finite-window resource bound}

The analytical results above are primarily formulated in the perfect
time-filter limit $T\to\infty$, which provides a convenient setting for
characterizing the long-time frame potential. In realistic
implementations, however, the evolution times must be sampled from a
finite window. Numerical results further show that the design-forming
Pauli and Clifford protocols can already approach their long-time
behavior within finite sampling windows; see \aref{sec:finite_time_numerics}. This motivates a more general question:
how does the temporal sampling resource depend on the number of
intermediate perturbations? We now derive a universal finite-window
constraint that quantifies how the number of independently sampled
evolution intervals limits the time window required for randomness
generation.

 Consider
$V_{\boldsymbol t} 
    =e^{-iHt_{m}}U e^{-iHt_{m-1}}\cdots Ue^{-iHt_1}$,
with $m$ evolution intervals separated by $(m-1)$ identical deterministic
perturbations $U$. Unlike the space--depth bound above, the following
argument does not rely on Hamiltonian locality, spectral nonresonance,
eigenbasis randomness, or randomness in the perturbations.

The key observation is that two independently sampled trajectories
necessarily make a large contribution to the frame potential whenever
their evolution times are sufficiently close in every interval. To
quantify the probability of such near-coincident trajectories, we
coarse-grain the $m$-dimensional time domain $[0,T]^m$. Specifically,
we divide each time coordinate into
$q=\lceil2m\sigma_HT\rceil$ bins, producing $q^m$ cells in total.
If two independently sampled time vectors $\boldsymbol t$ and
$\boldsymbol s$ fall into the same cell, then their evolution times
differ only slightly in every interval. Using unitary invariance of the
Hilbert--Schmidt norm and replacing the evolution factors one at a time
gives
\begin{align}
    \|V_{\boldsymbol t}-V_{\boldsymbol s}\|_2
    &\le
    \sqrt D\,\sigma_H
    \sum_{j=1}^{m}|t_j-s_j|
    \le \frac{\sqrt D}{2},
    \notag\\
    |\tr(V_{\boldsymbol s}^{\dagger}V_{\boldsymbol t})|
    &\ge
    D-\frac12
    \|V_{\boldsymbol t}-V_{\boldsymbol s}\|_2^2
    \ge \frac{7D}{8}.
    \label{eq:window_near_collision}
\end{align}
Thus each same-cell pair contributes at least an order-$D^{2k}$ term to
the frame potential. Since two independent time vectors fall into the
same cell with probability $q^{-m}$, retaining only these nonnegative
near-collision contributions yields
\begin{equation}
    F_{m}^{(k)}(T)
    \ge
    \frac{(7D/8)^{2k}}{q^{m}}
    \ge
    \frac{(7D/8)^{2k}}
    {\left[1+2m\sigma_HT\right]^{m}}.
    \label{eq:window_finite_bound}
\end{equation}
The $T^{-m}$ dependence of the lower bound follows from the
probability that two independently sampled time vectors lie in
the same grid cell. Each additional time coordinate contributes
a factor $q^{-1}$ to this probability.

Consequently, achieving
$\delta F_m^{(k)}(T)\le\varepsilon$ necessarily requires
\begin{equation}
    \sigma_HT
    =
    \Omega_{m,k,\varepsilon}
    \left(D^{2k/m}\right),
    \label{eq:window_minimum_bound}
\end{equation}
as $D\to\infty$ at fixed $m$, $k$, and $\varepsilon$. Increasing the
perturbation depth therefore relaxes the necessary sampling-window
scaling by introducing additional independent temporal coordinates.

This is a necessary constraint rather than a sufficient condition for
design formation. Perturbations that fail to mix the relevant
dynamical sectors may prevent convergence even when the bound on $T$
is satisfied. Conversely, the lower bound alone does not imply that
the actual finite-time correction decays as $T^{-m}$ in a given model.

The numerical results in \fref{fig:finite_window_scaling} show that the
depth dependence of the analytical bound is realized over a broad regime.
For GUE dynamics, the relative frame-potential
excess follows the reference scalings $T^{-m}$ closely for
$m=3,4,5,6$, with increasing $m$ producing a systematically faster
suppression of the finite-window correction. The same hierarchy is
observed for the Rydberg dynamics: despite its structured local
Hamiltonian, the initial decay again closely follows $T^{-m}$ before
crossing over to a nonzero long-time plateau. This plateau is consistent
with the locality-induced obstruction discussed in the universal
space--depth bound above: the fixed single-qubit perturbation has limited
spatial support, so for finite $m$ the available control area may remain
insufficient to fully overcome the local structure of the Hamiltonian.
Indeed, for the fixed $Z_1$ perturbation used here,
$Z_1H_{\mathrm{Ryd}}Z_1-H_{\mathrm{Ryd}}=-\Omega_1X_1$,
so the cumulative Hamiltonian variation is exactly
$\Lambda=(m-1)|\Omega_1|$, independent of the sampling window $T$.
Increasing $T$ therefore improves temporal sampling without enlarging
the energy-space localization window set by the fixed local control.
Thus, the finite-window and space--depth bounds capture two distinct
constraints: the former controls the approach to the long-time plateau,
whereas the latter constrains the plateau itself for local dynamics.
The observed scaling before the plateau is consistent with the depth-dependent
exponent of the universal lower bound in
\eref{eq:window_finite_bound} for
both random-matrix and local many-body examples over the resolved
scaling regime.

This separation of spatial and temporal constraints also clarifies the
relation of our finite-window bound to recent results on Hamiltonian-generated
randomness. A related depth-dependent finite-time scaling was reported in
\rcite{hou2026state} for state frame potentials, where an $m$-step
Hamiltonian protocol yields a correction scaling as $\sim T^{-m}$.
Here we establish the corresponding necessary resource scaling for unitary
ensembles without assuming additive nonresonance, eigenbasis averaging, or
any specific structure or randomness of the intermediate perturbations.
The coarse-grained collision picture further provides a direct physical
interpretation of the origin of this depth dependence.

\section{Summary}

In this work, we showed that unitary $k$-designs can be generated from
a single chaotic Hamiltonian when the evolution is interrupted by
suitable intermediate unitary perturbations. We derived a general
relation between the frame potential of the resulting ensemble and
that of the intermediate ensemble, identifying a broad class of
admissible perturbations. In particular, the intermediate ensemble
need not itself approximate a unitary design, provided that its trace moments
and frame-potential growth are sufficiently suppressed. Nontrivial
Pauli and Clifford ensembles provide representative examples, whereas
ensembles whose cardinality remains independent of the Hilbert-space
dimension generally fail. We further established, at the level of
frame potentials, a systematic correspondence in which temporal
protocols involving multiple independently controlled chaotic
Hamiltonians can be replaced by one-Hamiltonian protocols interleaved
with repeated fixed trace-suppressed perturbations. This control
substitution applies to any fixed perturbation satisfying the
trace-suppression condition; it does not rely on special choices such
as Pauli pulses and preserves asymptotic design formation while
replacing additional Hamiltonian controls with a single repeatable
operation.

Beyond design formation itself, we identified complementary resource
constraints on perturbation-assisted randomness generation. For local
Hamiltonians, the spatial support and perturbation depth obey a
space--depth lower bound, with their product providing a natural
control-resource scale that is also supported by the Rydberg
numerics. More generally, a finite sampling window is constrained by
coarse-grained temporal collisions, yielding a universal lower bound whose depth dependence is supported
numerically in both GUE and Rydberg dynamics over the resolved
scaling regime. Thus, increasing perturbation depth can reduce both
the spatial burden of local control and the temporal window required
to approach Haar-level randomness, although the two constraints have
different ranges of applicability.

Although the analytical treatment uses a Haar-eigenbasis average to
isolate a tractable random-matrix limit, the Rydberg results provide
complementary finite-size evidence that the perturbation-assisted
mechanism extends to structured local many-body dynamics.

At the trajectory level, the perturbations redirect the restricted temporal
orbit into increasingly distinguishable histories. In the process-tensor
picture, different intermediate operations generate distinct histories in
butterfly space whose associated unitary trajectories become increasingly
orthogonal under chaotic dynamics
\cite{Dowling2024operation,Donovan2026diagnos}. For the temporal ensembles
considered here, this separation is quantified by the frame potential, so
repeated fixed perturbations act as catalysts that expand the orbit of a
single chaotic Hamiltonian toward unitary-design formation. The additive
nonresonance condition underlying this mechanism is closely related to
quantum ergodicity and mixing
\cite{von2010proof,zhang2016ergodicity}; whether it admits a more quantitative
connection to higher levels of the ergodic hierarchy, including quantum
Bernoulli dynamics \cite{Frigg2025ergodic}, remains an open question.

We also comment on two closely related contemporaneous
studies~\cite{nandy2026unitary,hou2026state}, which were developed
independently of the present work. All three works were posted
within a short time interval. Ref.~\cite{nandy2026unitary} analyzes
specific Pauli-kick constructions for unitary-design generation,
with its one-kick result for GUE Hamiltonians agreeing with
\eref{eq:fixed_traceless_frame}. It also investigates
finite-temperature frame potentials and SYK models. Ref.~\cite{hou2026state} focuses primarily on state-design formation under
closely related Hamiltonian dynamics. By contrast, the present work develops an ensemble-level theory of
unitary-design formation for general intermediate ensembles. In particular, it provides criteria distinguishing design-forming
from design-limiting ensembles, establishes a fixed-perturbation construction
that is not restricted to Pauli kicks, and derives complementary spatial and
temporal resource constraints through the local space--depth and universal
finite-window bounds.

Several directions remain open. It would be useful to determine the
tightness of the locality-induced space--depth bound and to identify
local-control protocols that approach it. Likewise, understanding when
the universal finite-window bound is saturated beyond the models studied
here, particularly in the presence of symmetries and near-degeneracies,
would further clarify the spatial and temporal resources required for
design formation. Another direction is to establish quantitative
connections between the frame potentials of the ensembles generated by
multiple rounds of fixed perturbations and chaos diagnostics, including
dynamical entropy, process-tensor spatiotemporal entanglement, and
higher-order OTOCs
\cite{Roberts2017chaos,Foini2019eigenstate,Pappalardi2022eth,Dowling2024operation,Fava2025design,Pappalardi2025full,Donovan2026diagnos}.
Experimental realization on quantum-simulation platforms, together with
applications requiring higher-order unitary designs---such as classical
shadow estimation, R\'enyi-entropy measurements, and randomized
benchmarking---provides a further avenue for future work
\cite{HuangKuengPreskill2020,helsen2023thrifty,chen2024nonstab,
yang2025high,Dankert2009exact,elben2018renyi,vermersch2018unitary,
Nakata2021quantum}.

\section*{Acknowledgements}
This work was supported by the National Natural Science Foundation of China (92365202, 12475011, 11921005), 
the National Key R\&D Program of China (2024YFA1409002), 
the Shanghai Municipal Science and Technology Major Project (2019SHZDXZX01), 
the Shanghai Municipal Science and Technology Project (2SLZ601100), 
the Innovation Program for Quantum Science and Technology (2021ZD0302100).
\section*{Author contributions}
Z.Y. conducted the analytical and numerical calculations and prepared the initial draft of the manuscript. B.W. supervised the research. Both authors contributed to the conceptual development of the work, interpretation of the results, and revision of the manuscript. GPT-6-Astra was used to assist with writing code for numerical simulations and checking mathematical proofs. The authors reviewed the AI-assisted outputs and take full responsibility for the content of the manuscript.

\bibliographystyle{quantum}
\bibliography{ref}

\clearpage
\onecolumn
\appendix
\setlength{\abovedisplayskip}{8pt plus 2pt minus 3pt}
\setlength{\belowdisplayskip}{8pt plus 2pt minus 3pt}
In these appendices, we provide the technical derivations
supporting the results in the main text. We first establish the long-time
reduction under the $k$-th order additive nonresonance condition, including
the symmetry factors required when energy indices are repeated. We then
combine this reduction with the Haar-twirling expansion in
\aref{sec:app_prove_thm_main} to prove \thref{thm:main_oneH_int}, including the suppression of
subleading permutation sectors and the classification of the direct and
crossed pairings. We subsequently derive the frame potential of the
unperturbed single-Hamiltonian temporal ensemble and prove the monotonicity
of the frame potential under repeated fixed perturbations. Building on
these results, we prove \thref{thm:three_round_fixed_perturbation}, which
establishes Haar-level frame potentials after two fixed perturbations in
the asymptotic regime. We then prove
\thref{thm:local_space_depth_bound} by relating the residual energy overlap
to the excess frame potential and bounding how strongly perturbations of
limited spatial support can modify a local Hamiltonian. We then describe
the numerical sampling procedure in \aref{sec:numerical_sampling}. Finally, we provide
additional numerical results on finite-time convergence, showing that
design-forming intermediate ensembles approach their long-time frame
potentials within finite sampling windows and complementing the
finite-window resource analysis in the main text.

\section{Time average under
\texorpdfstring{$k$-th order additive nonresonance}{k-th order additive nonresonance}}
\label{sec:app_time_ave}

Here we derive the permutation-sector reduction used in the main text. Such
additive nonresonance assumptions are standard in long-time averaging
\cite{von2010proof,zhang2016ergodicity,Roberts2017chaos,Mark2024,Zhou2026three,Zhou2026realizing,sun2026unitary}.

To organize the $k$-th moment, collect the energy-eigenstate indices into
\begin{align*}
    \vec m
    &=(m_1,\ldots,m_k,\tilde m_1,\ldots,\tilde m_k)
      =\bfm\oplus\tilde{\bfm},\\
    \vec n
    &=(n_1,\ldots,n_k,\tilde n_1,\ldots,\tilde n_k)
      =\bfn\oplus\tilde{\bfn},
\end{align*}
where $m_i,n_i\in\{1,2,\ldots,D\}$ and $\bfm,\bfn\in[D]^k$.
For $\bfm\in[D]^k$, let $\mu_a(\bfm)$ denote the number of occurrences of
the label $a$, and define the symmetry factor
\[
    \rmS(\bfm):=\prod_{a=1}^{D}\mu_a(\bfm)!.
\]
Equivalently, $\rmS(\bfm)$ is the number of permutations in $S_k$ that leave
$\bfm$ invariant; for example, $\rmS(1,1,2,2,2)=2!3!=12$. This factor removes
the overcounting that arises when the permutation representation of an
additive-energy solution is not unique because some indices coincide.

\begin{proposition}
\label{prp:time_averaged_permutation_sector}
Consider the two unitaries defined in
\eref{eq:two_unitaries_convention}. Suppose that $H$ satisfies the
$k$-th order additive nonresonance condition in
\eref{eq:k_nonresonance_condition}. Then, in the perfect time-filter limit
$T\to\infty$,
\begin{equation}
    \bbE_{\delta t_1,\delta t_2}
    \left|\tr(\mathcal V^{\dagger}\mathcal U)\right|^{2k}
    =
    \sum_{\pi,\sigma\in S_k}
    \rmR_{\pi,\sigma}(P,Q),
    \label{eq:time_averaged_R_pi_sigma}
\end{equation}
where
\begin{align}
    \rmR_{\pi,\sigma}(P,Q)
    &:=
    \sum_{\bfm,\bfn\in[D]^k}
    \frac{1}{\rmS(\bfm)\rmS(\bfn)}\prod_{r=1}^{k}\Bigl[
    P_{m_r n_r}P^{\dagger}_{n_{\sigma(r)}m_{\pi(r)}}
    \notag\\
    &\hspace{7.2em}{}
    \times Q^{\dagger}_{n_r m_r}Q_{m_{\pi(r)}n_{\sigma(r)}}\Bigr].
    \label{eq:R_pi_sigma_def}
\end{align}
\end{proposition}

The quantity $\rmR_{\pi,\sigma}(P,Q)$ separates the spectral selection imposed
by the long-time average, encoded by $(\pi,\sigma)$, from the intermediate
ensemble, which enters only through the joint statistics of the independent
draws $P$ and $Q$.

\begin{proof}[Proof of \pref{prp:time_averaged_permutation_sector}]
For the two unitaries in \eref{eq:two_unitaries_convention}, cyclicity of the
trace gives
\begin{align*}
    \left|\tr(\mathcal V^\dagger \mathcal U)\right|^2
    &=\sum_{m_1,n_1}\sum_{\tilde m_1,\tilde n_1}
    P_{m_1n_1}Q^\dagger_{n_1m_1}
    P^\dagger_{\tilde n_1\tilde m_1}Q_{\tilde m_1\tilde n_1}\\
    &\quad\times
    \rme^{-\rmi\delta t_1(E_{n_1}-E_{\tilde n_1})}
    \rme^{-\rmi\delta t_2(E_{m_1}-E_{\tilde m_1})}.
\end{align*}
Here all four indices label eigenstates of \(H\). Raising this expression to
the \(k\)-th power and using the index vectors introduced above yields
\begin{align}
    \left|\tr(\mathcal V^\dagger \mathcal U)\right|^{2k}
    &=\sum_{\vec m,\vec n\in[D]^{2k}}
    \exp\!\left[-\rmi\delta t_2\sum_{j=1}^{k}
    (E_{m_j}-E_{\tilde m_j})\right]
    \notag\\
    &\quad\times
    \exp\!\left[-\rmi\delta t_1\sum_{j=1}^{k}
    (E_{n_j}-E_{\tilde n_j})\right]
    \notag\\
    &\quad\times
    \prod_{r=1}^{k}P_{m_rn_r}Q^\dagger_{n_rm_r}
    P^\dagger_{\tilde n_r\tilde m_r}Q_{\tilde m_r\tilde n_r}.
    \label{eq:trUV_2k}
\end{align}
In the perfect time-filter limit, the independent averages over
\(\delta t_1\) and \(\delta t_2\) suppress all nonzero frequencies. A term in
\eref{eq:trUV_2k} can therefore survive only if
\begin{equation}
    \sum_{j=1}^{k}(E_{m_j}-E_{\tilde m_j})
    =\sum_{j=1}^{k}(E_{n_j}-E_{\tilde n_j})=0.
    \label{eq:energy_sum_constraint}
\end{equation}
The additive nonresonance condition implies that the complete set of solutions
is
\begin{align*}
    (\tilde m_1,\ldots,\tilde m_k)
    &=\pi(m_1,\ldots,m_k),\\
    (\tilde n_1,\ldots,\tilde n_k)
    &=\sigma(n_1,\ldots,n_k),
\end{align*}
for some \(\pi,\sigma\in S_k\). If an index vector contains repeated entries,
some permutations leave it invariant, so summing over the full permutation
groups overcounts each term by \(\rmS(\bfm)\rmS(\bfn)\). Substituting
these matchings into \eref{eq:trUV_2k} and dividing by the symmetry factors
gives \eref{eq:time_averaged_R_pi_sigma} and \eref{eq:R_pi_sigma_def}.
\end{proof}

\section{Proof of \thref{thm:main_oneH_int}}\label{sec:app_prove_thm_main}

For the Haar average, write
$\Phi_q(X)=\int_{\rmU(D)}\rmd W\,W^{\dagger\otimes q}XW^{\otimes q}$.
The permutation convention is
$V_\tau|v_1\>\otimes\cdots\otimes|v_q\>
=|v_{\tau^{-1}(1)}\>\otimes\cdots\otimes|v_{\tau^{-1}(q)}\>$.
The Weingarten expansion and its fixed-$q$, large-$D$ scaling are
\begin{equation}
    \Phi_q(X)=\sum_{\alpha,\beta\in S_q}
    \mathrm{Wg}^{(q)}_D(\alpha^{-1}\beta)
    \tr(V_\alpha^{\dagger}X)V_\beta,
    \label{eq:exact_schur_weyl_twirl}
\end{equation}
\begin{equation}
    \mathrm{Wg}^{(q)}_D(\tau)=\caO_q(D^{-q-|\tau|}),
    \qquad |\tau|=q-c(\tau),
    \label{eq:wg_scaling_schur}
\end{equation}
where $c(\tau)$ is the number of cycles of $\tau$.

In this section, we prove \thref{thm:main_oneH_int}.  We start by introducing a
piece of notation that will be used throughout the proof and then present a
corollary and several propositions to prepare for the proof. Expanding the four
operator factors at each replica in \eref{eq:R_pi_sigma_def} produces \(4k\)
matrix elements. To organize them compactly and systematically, define the
label set
\begin{equation}
    \Omega=
    \{a_r,b_r,c_r,d_r:r=1,\ldots,k\},\qquad \left|\Omega\right|=4k.
\end{equation}
For each \(\omega\in\Omega\), we associate an operator \(B^{(\omega)}\) and two
external indices \(x_\omega,y_\omega\) as follows:

\[
\begin{array}{c|c|c|c}\
\omega & B^{(\omega)} & x_\omega & y_\omega\\
\hline
a_r & P & m_r & n_r\\
b_r & Q^\dagger & n_r & m_r\\
c_r & P^\dagger & n_{\sigma(r)} & m_{\pi(r)}\\
d_r & Q & m_{\pi(r)} & n_{\sigma(r)}
\end{array}
\]
This notation allows us to rewrite the product of matrix elements in
\(\eref{eq:R_pi_sigma_def}\) as a product over the labels
\(\omega\in\Omega\). We now rewrite \eref{eq:R_pi_sigma_def} as 
\begin{equation}\label{eq:RAlternative}
       \rmR_{\pi,\sigma}(P,Q)
    =
    \sum_{\bfm,\bfn \in [D]^k}\frac{1}{\rmS(\bfm) \rmS(\bfn)}
    \prod_{\omega\in\Omega}
    B^{(\omega)}_{x_\omega y_\omega}.
\end{equation}
We obtain the following results.

\begin{corollary}\label{cor:app_frame_potential_expanded}
    Suppose that \(H\) satisfies the \(k\)-th order additive nonresonance
    condition in \(\eref{eq:k_nonresonance_condition}\).  Consider the
    ensemble \(\caE_{\scrN}(H,\caE_{\mathrm{int}})\) with intermediate
    ensemble \(\caE_{\mathrm{int}}\).  In the perfect time-filter limit
    \(T\to\infty\), the Haar-eigenbasis-averaged frame potential satisfies
    \begin{equation}
        \overline{F}_{\caE_{\scrN}(H,\caE_{\mathrm{int}})}^{(k)}
        =
        \sum_{\pi,\sigma\in S_k}\overline{\rmR}_{\pi,\sigma},
        \label{eq:haar_averaged_frame_potential_R_sum}
    \end{equation}
    where
    \begin{equation}
        \overline{\rmR}_{\pi,\sigma}
        =
        \sum_{\alpha,\beta \in S_{4k}}
        \mathrm{Wg}^{(4k)}_D(\alpha^{-1}\beta)\,
        \bbE_{P,Q\sim \caE_{\mathrm{int}}}
        \bigl[
            C_{\alpha}(P,Q)
        \bigr]\,
        \rmG_{\pi,\sigma}(\beta),
        \label{eq:haar_averaged_R_pi_sigma_weingarten}
    \end{equation}
    with
\begin{equation}
    C_\alpha(P,Q)
    =
    \tr\left\{
    \left[
    \bigotimes_{\omega\in\Omega}B^{(\omega)}
    \right]
    V_\alpha^{\dagger}
    \right\},
    \label{eq:C_alpha_def}
\end{equation}
and
\begin{equation}
    \rmG_{\pi,\sigma}(\beta)
    =
 \sum_{\bfm,\bfn \in [D]^k}\frac{ \<\mathbf x(\bfm,\bfn)|
    V_\beta
    |\mathbf y(\bfm,\bfn)\>}{\rmS(\bfm)\rmS(\bfn)}.
    \label{eq:N_beta_def}
\end{equation}
\end{corollary}

\begin{proposition}\label{prp:lower_trace_suppression}
If an ensemble $\caE$ satisfies \eref{eq:small_frame_int} at some
fixed order $k\ge1$ with exponent $\epsilon_k>0$ independent of $D$,
then it satisfies the criterion at every fixed order $\ell\ge1$.
Specifically,
\begin{equation}
    F_{\caE}^{(\ell)}=\caO_{k,\ell}(D^{2\ell-\epsilon_\ell}),
    \qquad
    \epsilon_\ell=\epsilon_k\min\{1,\ell/k\}.
\end{equation}
\end{proposition}

\begin{proposition}\label{prp:A_kD_upp}
    For $\bfm \in [D]^k$ and the symmetry factor $\rmS(\bfm)$ defined in
    \aref{sec:app_time_ave}, we have
    \begin{equation}\label{eq:A_kD_expree}
        A_k(D):=\sum_{\bfm \in [D]^k} \frac{1}{\rmS(\bfm)}=k!
\sum_{\substack{\mu_1,\ldots,\mu_D\ge 0\\
\mu_1+\cdots+\mu_D=k}}
\prod_{a=1}^D \frac{1}{(\mu_a!)^2}=D^{k}+\caO_k(D^{k-1}).
    \end{equation}
\end{proposition}

\begin{proposition}\label{prp:N_beta_upperbound}
Consider \(\rmG_{\pi,\sigma}(\beta)\) defined in \eref{eq:N_beta_def}. If $V_\beta$ perfectly matches all the indices of $\bfx(\bfm,\bfn)$ and $\mathbf y(\bfm,\bfn)$,
\begin{equation}
    \langle x(m,n)|V_\beta|y(m,n)\rangle\equiv 1
\quad
\forall m,n\in[D]^k.
\end{equation}
Then we have
\begin{equation}
    \rmG_{\pi,\sigma}(\beta)=D^{2k}+\caO_k(D^{2k-1}).
\end{equation}
If $V_\beta$ fails to perfectly match all the indices of $\bfx(\bfm,\bfn)$ and $\mathbf y(\bfm,\bfn)$, then
\begin{equation}
    \rmG_{\pi,\sigma}(\beta)=\caO_k(D^{2k-1}).
\end{equation}
\end{proposition}

\begin{proposition}\label{prp:C_leading_order}
    Suppose that the intermediate ensemble $\caE_{\mathrm{int}}$ satisfies the
    $k$-th order trace-suppression condition in \eref{eq:small_tr_ensemble}.
    Let $C_\alpha(P,Q)$ be defined as in \eref{eq:C_alpha_def}. If the permutation
    $\alpha$ is not a product of $2k$ disjoint transpositions, then
    \begin{equation}
        \left|
        \bbE_{P,Q\sim \caE_{\mathrm{int}}}
        \left[C_{\alpha}(P,Q)\right]
        \right|
        \le
        \caO_k\left(D^{2k-\epsilon}\right),
        \qquad
        \epsilon=\Omega_k(1),
    \end{equation}
    where $\caO_k(\cdot)$ and $\Omega_k(\cdot)$ denote the upper and lower bounds of the scaling with respect to dimension $D$ for a fixed $k$.
\end{proposition}

\begin{lemma}\label{lem:SUMIDENTITY} For any $k\ge 1$ and any parameter $y$, one has \begin{equation} \sum_{\eta\in S_k} y^{|\mathrm{Fix}(\eta)|} = k!\sum_{r=0}^{k}\frac{(y-1)^r}{r!}. \end{equation} \end{lemma}

The proof of \coref{cor:app_frame_potential_expanded} is obtained by
substituting the exact Weingarten expansion
\(\eref{eq:exact_schur_weyl_twirl}\) into the time-averaged expression
\(\eref{eq:time_averaged_R_pi_sigma}\). The proofs of Propositions~\ref{prp:lower_trace_suppression}, \ref{prp:A_kD_upp}, \ref{prp:N_beta_upperbound} and \ref{prp:C_leading_order} are presented in Secs.~\ref{sec:app_prove_lower_trace}, \ref{sec:app_prove_A_kD}, \ref{sec:app_prove_N_beta_upp} and \ref{sec:app_prove_C_lead} respectively. The proof of \lref{lem:SUMIDENTITY} is presented in \sref{sec:app_prove_SUMIDENTITY}.

\subsection{Proof of \coref{cor:app_frame_potential_expanded}}
\begin{proof}[Proof of \coref{cor:app_frame_potential_expanded}]
    The average of $\rmR_{\pi,\sigma}(P,Q)$ over a Haar-random basis is 
    \begin{align}\label{eq:RAVEINT}
        \bbE_W \rmR_{\pi,\sigma}(P,Q)&=\bbE_W \sum_{\bfm,\bfn\in[D^k]}\frac{1}{\rmS(\bfm)\rmS(\bfn)}\tr\Bigg[\left(P\otimes Q^{\dagger} \otimes P^{\dagger} \otimes Q\right)^{
        \otimes k}\bigotimes_{r=1}^{k}\Big(W^{\dagger}|m_r\>\<n_r|W\otimes W^{\dagger}|n_r\>\<m_r|W \notag \\
        &\quad \otimes W^{\dagger}|n_{\sigma(r)}\>\<m_{\pi(r)}|W\otimes  W^{\dagger}|m_{\pi(r)}\>\<n_{\sigma(r)}|W\Big)\Bigg] \notag\\
        &=\sum_{\bfm,\bfn\in[D^k]}\frac{1}{\rmS(\bfm)\rmS(\bfn)}\tr\Bigg[\left(P\otimes Q^{\dagger} \otimes P^{\dagger} \otimes Q\right)^{\otimes k}\Phi_{4k}\Bigg[\bigotimes_{r=1}^{k}\Big(|m_r\>\<n_r| \otimes|n_r\>\<m_r| \notag\\
        &\quad \otimes |n_{\sigma(r)}\>\<m_{\pi(r)}| \otimes |m_{\pi(r)}\>\<n_{\sigma(r)}|\Big)\Bigg].
    \end{align}
Substituting \eqsref{eq:RAlternative}{eq:exact_schur_weyl_twirl} into \eref{eq:RAVEINT}, we have
\begin{equation}
    \bbE_W \rmR_{\pi,\sigma}(P,Q)= 
        \sum_{\alpha,\beta \in S_{4k}}
        \mathrm{Wg}^{(4k)}_D(\alpha^{-1}\beta)\,
            C_{\alpha}(P,Q)
        \rmG_{\pi,\sigma}(\beta),
\end{equation}   with
\begin{equation}
    C_\alpha(P,Q)
    =
    \tr\left\{
    \left[
    \bigotimes_{\omega\in\Omega}B^{(\omega)}
    \right]
    V_\alpha^{\dagger}
    \right\},
\end{equation}
and
\begin{equation}
    \rmG_{\pi,\sigma}(\beta)
    =
 \sum_{\bfm,\bfn \in [D]^k}\frac{ \<\mathbf x(\bfm,\bfn)|
    V_\beta
    |\mathbf y(\bfm,\bfn)\>}{\rmS(\bfm)\rmS(\bfn)}.
\end{equation}
 Therefore,
     \begin{equation}
        \overline{F}_{\caE_{\scrN}(H,\caE_{\mathrm{int}})}^{(k)}
        =\sum_{\pi,\sigma\in S_k}\bbE_{P,Q\sim \caE_{\text{int}}} \bbE_W \rmR_{\pi,\sigma}(P,Q) =
        \sum_{\pi,\sigma\in S_k}\overline{\rmR}_{\pi,\sigma},
    \end{equation}
with
\begin{equation}
        \overline{\rmR}_{\pi,\sigma}
        =
        \sum_{\alpha,\beta \in S_{4k}}
        \mathrm{Wg}^{(4k)}_D(\alpha^{-1}\beta)\,
        \bbE_{P,Q\sim \caE_{\mathrm{int}}}
        \bigl[
            C_{\alpha}(P,Q)
        \bigr]\,
        \rmG_{\pi,\sigma}(\beta),
    \end{equation}
    which proves \coref{cor:app_frame_potential_expanded}.
\end{proof}

\subsection{Proof of \pref{prp:lower_trace_suppression}}
\label{sec:app_prove_lower_trace}

\begin{proof}[Proof of \pref{prp:lower_trace_suppression}]
Fix integers $k,\ell\ge 1$, and let $U$ and $V$ be independent draws
from $\caE$. Since $U^\dagger V$ is unitary, $|\tr(U^\dagger V)|\le D$.
Thus the random variable $q:=D^{-2}|\tr(U^\dagger V)|^2$ takes values
in $[0,1]$. By the definition of the frame potential,
\begin{equation}
    r_j:=\frac{F_{\caE}^{(j)}}{D^{2j}}=\bbE[q^j]
    \qquad\text{for every integer }j\ge 1.
\end{equation}

If $\ell\ge k$, then $q^\ell\le q^k$ pointwise, and hence
$r_\ell\le r_k$. If $1\le\ell<k$, the function $x\mapsto x^{\ell/k}$
is concave on $[0,1]$. Applying Jensen's inequality to $q^k$ gives
\begin{equation}
    r_\ell
    =\bbE\!\left[(q^k)^{\ell/k}\right]
    \le \left(\bbE[q^k]\right)^{\ell/k}
    =r_k^{\ell/k}.
\end{equation}
Consequently, $r_\ell\le r_k^{\min\{1,\ell/k\}}$ for all $\ell\ge 1$.

By hypothesis, there exist constants $C_k,\epsilon_k>0$, independent
of $D$, such that $r_k\le C_k D^{-\epsilon_k}$ for all sufficiently
large $D$. Setting $\alpha:=\min\{1,\ell/k\}$, we obtain
\begin{equation}
    F_{\caE}^{(\ell)}
    =D^{2\ell}r_\ell
    \le C_k^\alpha D^{2\ell-\alpha\epsilon_k}.
\end{equation}
Thus \eref{eq:small_frame_int} holds at order $\ell$ with exponent
$\epsilon_\ell=\epsilon_k\min\{1,\ell/k\}>0$, as claimed.
\end{proof}

\subsection{Proof of \pref{prp:A_kD_upp}}\label{sec:app_prove_A_kD}
\begin{proof}[Proof of \pref{prp:A_kD_upp}]
Since $\mu_i(\bfm)$ counts the number of occurrences of label $i$ in $\bfm$,
for any nonnegative integer vector
$(\mu_1,\ldots,\mu_D)$ satisfying $\mu_1+\cdots+\mu_D=k$, the number of
vectors $\bfm\in[D]^k$ with these occupation numbers is
\begin{equation}
    \frac{k!}{\mu_1!\mu_2!\cdots\mu_D!}.
\end{equation}
Therefore the total contribution of all vectors with the same occupation numbers is
\begin{equation}
    \frac{k!}{\prod_{a=1}^D\mu_a!}
\cdot
\frac{1}{\prod_{a=1}^D\mu_a!}
=
k!
\prod_{a=1}^D\frac{1}{(\mu_a!)^2}.
\end{equation}
Summing over all possible occupation-number vectors gives
\begin{equation}\label{eq:A_kD_EXP2}
    A_k(D)
=
k!
\sum_{\substack{\mu_1,\ldots,\mu_D\ge0\\
\mu_1+\cdots+\mu_D=k}}
\prod_{a=1}^D\frac{1}{(\mu_a!)^2}=k!
\sum_{p=1}^k
\binom{D}{p}
\sum_{\substack{\mu_1,\ldots,\mu_p\ge1\\
\mu_1+\cdots+\mu_p=k}}
\prod_{i=1}^p\frac{1}{(\mu_i!)^2},
\end{equation}
which proves the equality in \eref{eq:A_kD_expree}. The leading order of \eref{eq:A_kD_EXP2} comes from the case $p=k$, where the contribution reads
\begin{equation}
    k!\binom{D}{k}=
D(D-1)\cdots(D-k+1)=D^k+\caO_k(D^{k-1}).
\end{equation}
Contributions from all other cases are of order $\caO_k(D^{k-1})$, since only $\caO_k(D^{k-1})$ terms in the summation are involved. This completes the proof of \pref{prp:A_kD_upp}.
\end{proof}

\subsection{Proof of \pref{prp:N_beta_upperbound}}\label{sec:app_prove_N_beta_upp}
\begin{proof}[Proof of \pref{prp:N_beta_upperbound}]
We have
\begin{equation}
        \rmG_{\pi,\sigma}(\beta)
    =
    \sum_{\bfm,\bfn \in [D]^k}\frac{ \<\mathbf x(\bfm,\bfn)|
    V_\beta
    |\mathbf y(\bfm,\bfn)\>}{\rmS(\bfm) \rmS(\bfn)}
   .
\end{equation}
The matrix element
\begin{equation}
    \left\langle
    \mathbf x(\mathbf m,\mathbf n)
    \middle|
    V_\beta
    \middle|
    \mathbf y(\mathbf m,\mathbf n)
    \right\rangle
\end{equation}
is nonzero exactly when the indices of
\(\mathbf x(\mathbf m,\mathbf n)\) match those of
\(V_\beta\mathbf y(\mathbf m,\mathbf n)\).  Its
element is a product of Kronecker delta functions.  Hence its nonzero
condition imposes only equality constraints among the variables
\begin{equation}
    m_1,\ldots,m_k,n_1,\ldots,n_k .
\end{equation}

After all such equality constraints are imposed, suppose that
\(f_{\pi,\sigma}(\beta)\le 2k\) independent indices remain.  These free indices are
precisely the independent choices left in the summation over
\(\mathbf m,\mathbf n\).  Each independent free index can take \(D\) possible values. 

Therefore, if $V_\beta$ does not perfectly match the indices in $\bfx(\bfm,\bfn)$ and $\bfy(\bfm,\bfn)$ for every choice of $\bfm$ and $\bfn$, then $f_{\pi,\sigma}(\beta)\leq 2k-1$. Then $\rmG_{\pi,\sigma}(\beta)$ is upper-bounded as 
\begin{equation}
    \rmG_{\pi,\sigma}(\beta)=\sum_{\bfm,\bfn \in [D]^k}\frac{ \<\mathbf x(\bfm,\bfn)|
    V_\beta
    |\mathbf y(\bfm,\bfn)\>}{\rmS(\bfm) \rmS(\bfn)}\le\sum_{\bfm,\bfn \in [D]^k} \<\mathbf x(\bfm,\bfn)|
    V_\beta
    |\mathbf y(\bfm,\bfn)\>=D^{f_{\pi,\sigma}(\beta)}\le D^{2k-1}.
\end{equation}

If $V_\beta$ perfectly matches the indices, we have $ \<\mathbf x(\bfm,\bfn)|
    V_\beta
    |\mathbf y(\bfm,\bfn)\>=1$, which leads to
\begin{equation}
     \rmG_{\pi,\sigma}(\beta)=\sum_{\bfm,\bfn \in [D]^k}\frac{1}{\rmS(\bfm) \rmS(\bfn)}=\left[A_k(D)\right]^2=D^{2k}+\caO_k(D^{2k-1}),
\end{equation}
where the second equality follows from \pref{prp:A_kD_upp}. This proves \pref{prp:N_beta_upperbound}.
\end{proof}

\subsection{Proof of \pref{prp:C_leading_order}}\label{sec:app_prove_C_lead}

\begin{proof}[Proof of \pref{prp:C_leading_order}]
We first prove that if $\alpha$ has a fixed point, equivalently a
$1$-cycle, then its contribution is suppressed. 

Let
$s=s_p+s_q$ be the total number of fixed points of $\alpha$, where
$s_p$ and $s_q$ denote the numbers of fixed points associated with
$P$-type and $Q$-type factors, respectively. Using the cycle expansion
of permutation traces, we have
\begin{align}\label{eq:app_C_upp}
    \left|
    \mathbb{E}_{P,Q\sim\caE_{\mathrm{int}}}
    C_\alpha(P,Q)
    \right|
    &=
    \left|
    \mathbb{E}_{P,Q\sim\caE_{\mathrm{int}}}
    \prod_{c\in \mathrm{cyc}(\alpha)}
    \tr\left[
    \prod_{\omega\in c} B^{(\omega)}
    \right]
    \right|
    \notag\\
    &\le
    D^{|\mathrm{cyc}(\alpha)|-s}
    \mathbb{E}_{P\sim\caE_{\mathrm{int}}}
    \left|\tr(P)\right|^{s_p}
    \mathbb{E}_{Q\sim\caE_{\mathrm{int}}}
    \left|\tr(Q)\right|^{s_q}
    \notag\\
    &\le
    D^{|\mathrm{cyc}(\alpha)|-s}
    \mathbb{E}_{U\sim\caE_{\mathrm{int}}}
    \left|\tr(U)\right|^{s}.
\end{align}
Here the first inequality follows from
\[
    \left|
    \tr\left[
    \prod_{\omega\in c} B^{(\omega)}
    \right]
    \right|
    \le D
\]
for every nontrivial cycle $c$. The second inequality follows from H\"older's inequality and the
fact that $P$ and $Q$ are drawn from the same ensemble.

Since every cycle that is not a $1$-cycle has length at least $2$, the
number of nontrivial cycles satisfies
\begin{equation}\label{eq:app_nonsingle_upp}
    |\mathrm{cyc}(\alpha)|-s
    \le
    \left\lfloor\frac{4k-s}{2}\right\rfloor
    =
    \begin{cases}
        2k-\dfrac{s}{2}, & \text{if $s$ is even}, \\
        2k-\dfrac{s+1}{2}, & \text{if $s$ is odd}.
    \end{cases}
\end{equation}
On the other hand, by the  trace-suppression condition, we have
\begin{equation}\label{eq:app_trace_upp}
    \mathbb{E}_{U\sim\caE_{\mathrm{int}}}
    \left|\tr(U)\right|^s
    \le
    \begin{cases}
        \caO_k\left(D^{s/2-\epsilon}\right),
        & \text{if $s$ is even}, \\[4pt]
        \caO_k\left(D^{(s+1)/2-\epsilon}\right),
        & \text{if $s$ is odd},
    \end{cases}
\end{equation}
for some $\epsilon=\Omega_k(1)$. Combining
\eqsref{eq:app_C_upp}{eq:app_trace_upp} with
\eref{eq:app_nonsingle_upp}, we obtain
\begin{equation}
    \left|
    \mathbb{E}_{P,Q\sim\caE_{\mathrm{int}}}
    C_\alpha(P,Q)
    \right|
    \le
    \caO_k\left(D^{2k-\epsilon}\right).
\end{equation}
Therefore, any permutation with at least one fixed point cannot contribute
to the leading order $\caO(D^{2k})$.

It remains to consider permutations $\alpha$ with no fixed points but with at least one cycle of length greater than two. Since $\alpha$ has no fixed points, every cycle in its cycle decomposition has length at least two. The presence of at least one cycle of length greater than two then reduces the maximum possible number of cycles. More precisely, because $\alpha$ acts on $4k$ elements, the largest number of cycles under these assumptions is obtained by using cycles of length two whenever possible, except for the necessary longer-cycle structure. Hence
\begin{equation} 
|\mathrm{cyc}(\alpha)|\le 2k-1. 
\end{equation} 
Therefore, \begin{equation} 
\left| \mathbb{E}_{P,Q\sim\caE_{\mathrm{int}}} C_\alpha(P,Q) \right| \le D^{|\mathrm{cyc}(\alpha)|} \le D^{2k-1}, 
\end{equation} which is suppressed relative to the leading order $\caO(D^{2k})$. Consequently, a contribution of order $\caO(D^{2k})$ can arise only when
$\alpha$ consists of $2k$ disjoint transpositions. This proves \pref{prp:C_leading_order}.
\end{proof}

\subsection{Proof of \lref{lem:SUMIDENTITY}}
\label{sec:app_prove_SUMIDENTITY}

\begin{proof}[Proof of \lref{lem:SUMIDENTITY}]
Let $[k]=\{1,\ldots,k\}$. For a permutation $\eta\in S_k$, we write \begin{equation} y^{|\mathrm{Fix}(\eta)|} = \prod_{i=1}^{k} \left[1+(y-1)\delta_{\eta(i)=i}\right]. \end{equation}

We now expand this product.  In the expansion, for each position
$i\in[k]$ one chooses either the term $1$ or the term
$(y-1)\delta_{\eta(i)=i}$.  Equivalently, one chooses a subset
$A\subseteq[k]$ of positions from which the second term is selected.
This gives
\begin{equation}
    \prod_{i=1}^{k}
    \left[1+(y-1)\delta_{\eta(i)=i}\right]
    =
    \sum_{A\subseteq[k]}
    (y-1)^{|A|}
    \prod_{i\in A}\delta_{\eta(i)=i}.
\end{equation}
The product of indicators is equal to $1$ precisely when every element
of $A$ is fixed by $\eta$, and is equal to $0$ otherwise.  Summing over
all permutations, we obtain
\begin{align}
    \sum_{\eta\in S_k} y^{|\mathrm{Fix}(\eta)|}
    &=
    \sum_{\eta\in S_k}
    \sum_{A\subseteq[k]}
    (y-1)^{|A|}
    \prod_{i\in A}\delta_{\eta(i)=i}  \notag\\
    &=
    \sum_{A\subseteq[k]}
    (y-1)^{|A|}
    \sum_{\eta\in S_k}
    \prod_{i\in A}\delta_{\eta(i)=i} \notag \\
    &=
    \sum_{A\subseteq[k]}
    (y-1)^{|A|}
    \left|
    \{\eta\in S_k:\eta(i)=i \text{ for all } i\in A\}
    \right| .
\end{align}

It remains to count the permutations in the last set.  If all elements
of $A$ are fixed pointwise, then the permutation is completely
determined by its action on the complement $[k]\setminus A$.  Conversely,
any permutation of $[k]\setminus A$, together with the requirement that
all elements of $A$ are fixed, defines a unique permutation of $[k]$.
Thus
\begin{equation}
    \left|
    \{\eta\in S_k:\eta(i)=i \text{ for all } i\in A\}
    \right|
    =
    (k-|A|)! .
\end{equation}
Therefore,
\begin{equation}
    \sum_{\eta\in S_k} y^{|\mathrm{Fix}(\eta)|}
    =
    \sum_{A\subseteq[k]}
    (y-1)^{|A|}(k-|A|)! .
\end{equation}

Finally, we group subsets according to their size.  For each
$r=0,\ldots,k$, there are $\binom{k}{r}$ subsets $A\subseteq[k]$ with
$|A|=r$.  Hence
\begin{equation}
\begin{aligned}
    \sum_{\eta\in S_k} y^{|\mathrm{Fix}(\eta)|}
    &=
    \sum_{r=0}^{k}
    \binom{k}{r}(y-1)^r(k-r)!  \\
    &=
    \sum_{r=0}^{k}
    \frac{k!}{r!(k-r)!}(y-1)^r(k-r)!  \\
    &=
    k!\sum_{r=0}^{k}\frac{(y-1)^r}{r!}.
\end{aligned}
\end{equation}
This proves \lref{lem:SUMIDENTITY}.
\end{proof}

\subsection{Proof of \thref{thm:main_oneH_int}}
With the results above, we are ready to prove \thref{thm:main_oneH_int}.
\begin{proof}[Proof of \thref{thm:main_oneH_int}]
    First, we define the subset $M_{4k}$ of the permutation group $S_{4k}$ to be the set of all permutations consisting of $2k$ disjoint transpositions. By \eref{eq:wg_scaling_schur} and \pref{prp:N_beta_upperbound}, we have
    \begin{align}\label{eq:app_Ave_R}
         \overline{\rmR}_{\pi,\sigma}
        &=
        \sum_{\alpha,\beta \in S_{4k}}
        \mathrm{Wg}^{(4k)}_D(\alpha^{-1}\beta)\,
        \bbE_{P,Q\sim \caE_{\mathrm{int}}}
        \left[
            C_{\alpha}(P,Q)
        \right]\,
        \rmG_{\pi,\sigma}(\beta)\notag\\
        &=\sum_{\alpha\in M_{4k}}D^{-4k}
 \bbE_{P,Q\sim \caE_{\mathrm{int}}}
        \left[C_{\alpha}(P,Q)\right] \rmG_{\pi,\sigma}(\alpha) +o_{D\to\infty}(1).
        \end{align}
In order to give the leading order contribution one must impose $\rmG_{\pi,\sigma}(\alpha)=D^{2k}$, which indicates that $V_\alpha$ perfectly matches all the indices of $\mathbf x(\bfm,\bfn)$ and $\mathbf y(\bfm,\bfn)$.

\[
\begin{array}{c|c|c|c}\
\omega & B^{(\omega)} & x_\omega & y_\omega\\
\hline
a_r & P & m_r & n_r\\
b_r & Q^\dagger & n_r & m_r\\
c_r & P^\dagger & n_{\sigma(r)} & m_{\pi(r)}\\
d_r & Q & m_{\pi(r)} & n_{\sigma(r)}
\end{array}
\]

For each \(1\le r\le k\), the two occurrences of the index \(m_r\) in
\(\bfy\) must be mapped to the corresponding two occurrences of \(m_r\) in
\(\bfx\).  Equivalently, the following two sets of labels must be matched:
\begin{align}
    \left\{b_r,c_{\pi^{-1}(r)}\right\}
    &\longrightarrow
    \left\{a_r,d_{\pi^{-1}(r)}\right\},
    \label{eq:m_matching_set_app}
    \\
    \left\{a_r,d_{\sigma^{-1}(r)}\right\}
    &\longrightarrow
    \left\{b_r,c_{\sigma^{-1}(r)}\right\}.
    \label{eq:n_matching_set_app}
\end{align}
The same argument applies to the \(n_r\) indices. These constraints leave two
local perfect-pairing patterns as shown in \fref{fig:pairing_schematic}. 
\begin{figure}[!tbp]
    \centering
    \includegraphics[width=0.84\textwidth]{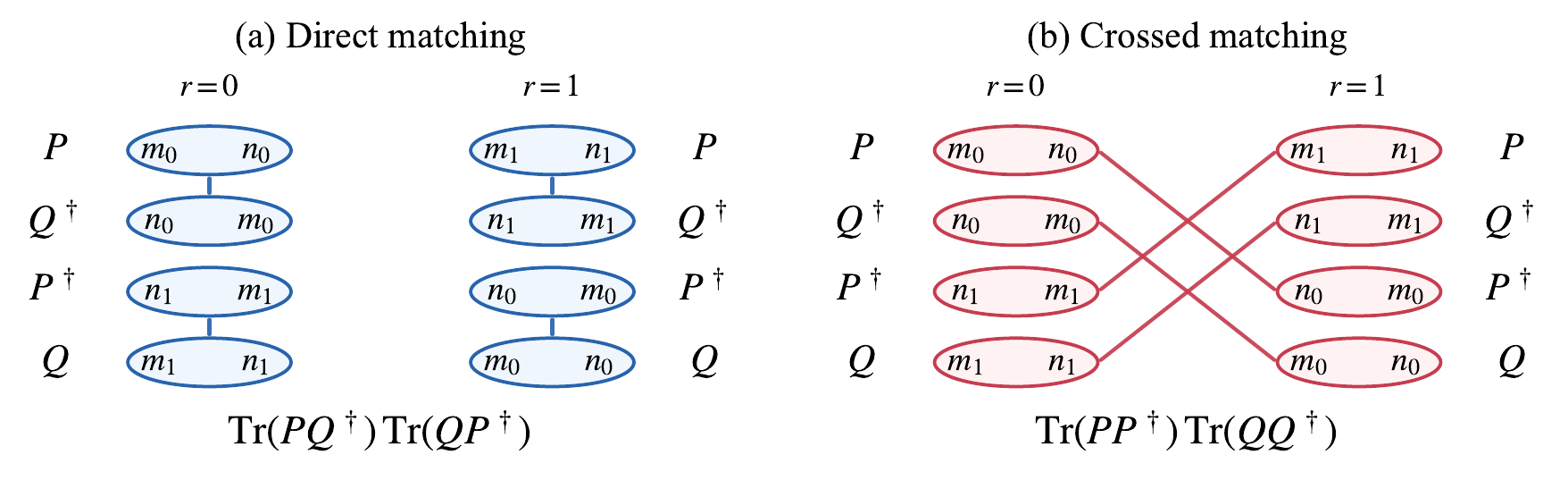}
    \caption{The two local perfect-matching structures for two replicas labelled
    $r=0,1$, with $\pi=\sigma=(0\,1)$ fixed in both panels. In each box, the
    left and right entries denote, respectively, the indices of
    $\mathbf{x}(\bfm,\bfn)$ and $\mathbf{y}(\bfm,\bfn)$ to be matched.
    (a) The direct matching connects indices within each replica and contributes
    the factor
    $\operatorname{Tr}(P Q^\dagger)\operatorname{Tr}(Q P^\dagger)$.
    (b) The crossed matching connects indices between different replicas and
    contributes the factor
    $\operatorname{Tr}(P P^\dagger)\operatorname{Tr}(Q Q^\dagger)$.}
    \label{fig:pairing_schematic}
\end{figure}
The direct matching
\begin{equation}
    b_r\leftrightarrow a_r,
    \qquad
    c_r\leftrightarrow d_r,
    \label{eq:app_direct_matching}
\end{equation}
yields a contribution of a factor $\left|\tr(Q^\dagger P)\right|^2$ to $C_{\alpha}(P,Q)$, according to the following identity
\begin{equation}
    \tr\left[(A\otimes B)\SWAP\right]=\tr(AB).
\end{equation}
Thus a direct pairing retains the overlap between the independently
sampled perturbations $P$ and $Q$ from the intermediate ensemble.

The crossed matching instead requires
\begin{align}
    b_r\leftrightarrow d_{\pi^{-1}(r)},
    \qquad
    a_r\leftrightarrow c_{\pi^{-1}(r)},
    \label{eq:app_matching_m}
    \\
    d_{\sigma^{-1}(r)}\leftrightarrow b_r,
    \qquad
    a_r\leftrightarrow c_{\sigma^{-1}(r)} .
    \label{eq:app_matching_n}
\end{align}

For a crossed pairing, each operator is contracted with its adjoint, giving
\begin{equation}
    \tr(Q^\dagger Q)\tr(P^\dagger P)=D^2 .
\end{equation}
The two relations in \eqsref{eq:app_matching_m}{eq:app_matching_n} must,
however, hold simultaneously. The same $b_r$ must be paired with both
$d_{\pi^{-1}(r)}$ and $d_{\sigma^{-1}(r)}$, and the same $a_r$ with both
$c_{\pi^{-1}(r)}$ and $c_{\sigma^{-1}(r)}$. Hence
\begin{equation}
    \pi^{-1}(r)=\sigma^{-1}(r).
\end{equation}
Equivalently, for $\eta=\sigma\pi^{-1}$,
\begin{equation}
    \eta(r)=\sigma\pi^{-1}(r)=r .
\end{equation}
A crossed pairing is therefore allowed only for $r\in\mathrm{Fix}(\eta)$.
Let $N_c$ be the number of crossed pairings and $N_d=k-N_c$ the number of
direct pairings, with $0\le N_c\le|\mathrm{Fix}(\eta)|$. After the leading
Weingarten factor and $\rmG_{\pi,\sigma}(\alpha)=D^{2k}$ are included, such a
configuration contributes
\begin{equation}
    D^{-2k}D^{2N_c}\bbE_{P,Q\sim\caE_{\mathrm{int}}}\left[\left|\tr(Q^\dagger P)\right|^{2N_d}\right]=D^{2(N_c-k)}\bbE_{P,Q\sim\caE_{\mathrm{int}}}\left[\left|\tr(Q^\dagger P)\right|^{2(k-N_c)}\right],
\end{equation}
The $N_c$ crossed replicas can be chosen from the allowed fixed points in
\begin{equation}
    \binom{\left|\mathrm{Fix}(\eta)\right|}{N_c}
\end{equation}
ways. Summing all perfect matchings in \eref{eq:app_Ave_R} therefore gives
\begin{align}
      \overline{\rmR}_{\pi,\sigma}&=\sum_{j=0}^{\left|\mathrm{Fix}(\eta)\right|}\binom{\left|\mathrm{Fix}(\eta)\right|}{j} D^{2(j-k)}\bbE_{P,Q\sim\caE_{\mathrm{int}}}\left[\left|\tr(Q^\dagger P)\right|^{2(k-j)}\right]\notag\\
      &=\bbE_{P,Q\sim\caE_{\mathrm{int}}}\left[\sum_{j=0}^{\left|\mathrm{Fix}(\eta)\right|}\binom{\left|\mathrm{Fix}(\eta)\right|}{j}q^{k-j}\right]=\bbE_{P,Q\sim\caE_{\mathrm{int}}}\left[\left(1+q\right)^{\left|\mathrm{Fix}(\eta)\right|}q^{k-\left|\mathrm{Fix}(\eta)\right|}\right],
\end{align}
where $q=D^{-2}\left|\tr(Q^\dagger P)\right|^2$. Finally,

\begin{align}
    \overline{F}_{\caE_{\scrN}(H,\caE_{\mathrm{int}})}^{(k)}
        &=
        \sum_{\pi,\sigma\in S_k}\overline{\rmR}_{\pi,\sigma}=\sum_{\pi,\sigma}\bbE_{P,Q\sim\caE_{\mathrm{int}}}\left[\left(1+q\right)^{\left|\mathrm{Fix}(\eta)\right|}q^{k-\left|\mathrm{Fix}(\eta)\right|}\right]+o_{D\to\infty}(1)\notag\\
        &=k!\sum_{\eta\in S_k}\bbE_{P,Q\sim\caE_{\mathrm{int}}}\left[\left(1+q\right)^{\left|\mathrm{Fix}(\eta)\right|}q^{k-\left|\mathrm{Fix}(\eta)\right|}\right]+o_{D\to\infty}(1)\notag\\
        &=(k!)^2\sum_{r=0}^{k}\bbE_{P,Q\sim\caE_{\mathrm{int}}}\left(\frac{q^{k-r}}{r!}\right)+o_{D\to\infty}(1)\notag\\
        &=(k!)^2\sum_{r=0}^k \frac{D^{2(r-k)}}{r!} F_{\caE_{\mathrm{int}}}^{(k-r)}+o_{D\to\infty}(1)\notag\\
        &= k!+\sum_{\ell=1}^{k}\frac{(k!)^2F_{\caE_{\mathrm{int}}}^{(\ell)}}{(k-\ell)!D^{2\ell}}+o_{D\to\infty}(1),
\end{align}
In the second line, the change of variables
$(\pi,\sigma)\mapsto(\eta,\sigma)$ makes $\sigma$ an independent label whose
sum contributes $k!$. The third line follows from \lref{lem:SUMIDENTITY}; the
fourth uses the definition of the frame potential, and the last sets
$\ell=k-r$ and uses $F_{\caE_{\mathrm{int}}}^{(0)}=1$.
\end{proof}

This final expression separates the two pairing sectors. The all-crossed
configuration has $N_c=k$, so it is possible only when every replica is a
fixed point of $\eta$, i.e., when $\eta=\mathrm{id}$ or $\pi=\sigma$. The
$k!$ compatible choices contract every perturbation with its own adjoint and
retain no intermediate-ensemble dependence.
If instead $\ell$ replicas are directly paired, their product is
$|\tr(Q^\dagger P)|^{2\ell}$, whose average is precisely
$F_{\caE_{\mathrm{int}}}^{(\ell)}$. The remaining $k-\ell$ crossed pairings
provide the dimension factors, while the compatible choices of replicas and
permutations give $(k!)^2/(k-\ell)!$. Together with the leading Weingarten
normalization, this yields the order-$\ell$ correction
$(k!)^2F_{\caE_{\mathrm{int}}}^{(\ell)}/[(k-\ell)!D^{2\ell}]$. Thus direct
pairings generate the complete hierarchy of corrections that depends on the
intermediate ensemble.

\section{Frame potential of
\texorpdfstring{$\caE_{\scrL(H)}$}{E\_L(H)}}\label{sm:SMFrameLH}
In this section we compute the frame potential $F_{\caE_{\scrL(H)}}$ of a fixed chaotic Hamiltonian \(H\) in the main text. Let
$ H=\sum_{a=1}^{D} E_a |a\>\<a|$ be the spectral decomposition of \(H\). The \(k\)-th frame potential reads
\begin{equation}
    F_{\caE_{\scrL(H)}}^{(k)}
    =
    \lim_{T\to\infty}
    \frac{1}{T^2}
    \int_0^T dt_1dt_2
    \left|
    \tr\left[e^{-iH(t_1-t_2)}\right]
    \right|^{2k},
    \label{eq:appFrameLHDef}
\end{equation}
which can be expanded as
\begin{align}
    F_{\caE_{\scrL(H)}}^{(k)}
    &=
    \sum_{a_1,\ldots,a_k=1}^{D}
    \sum_{b_1,\ldots,b_k=1}^{D}
    \lim_{T\to\infty}
    \frac{1}{T^2}
    \int_0^T dt_1dt_2\notag\\
    &\quad\exp\left[
    -i(t_1-t_2)
    \sum_{i=1}^{k}\left(E_{a_i}-E_{b_i}
    \right)
    \right].
    \label{eq:appFrameLHExpanded}
\end{align}
The only surviving terms in \eref{eq:appFrameLHExpanded} satisfy the energy conservation $\sum_{r=1}^k \left(E_{a_r}-E_{b_r}\right)=0$.
Therefore,
\begin{equation}
    F_{\caE_{\scrL(H)}}^{(k)}
    =
    \sum_{a_1,\ldots,a_k=1}^{D}
    \sum_{b_1,\ldots,b_k=1}^{D}
    \mathbf \delta
    \left[
    \sum_{i=1}^{k}E_{a_i}
    =
    \sum_{i=1}^{k}E_{b_i}
    \right],
    \label{eq:appFrameLHIndicator}
\end{equation}
which counts the number of pairings satisfying the energy-conservation
condition. 
For a  vector $\bfa=(a_1,\ldots,a_k)$, define its
occupation vector $\boldsymbol{\gamma}=(\gamma_1,\ldots,\gamma_D)$, with
\begin{equation}
        \gamma_j(\bfa)
    =
    \left|\{i:a_i=j\}\right|,
    \quad
    j=1,\ldots,D.
\end{equation}
counting the number of occurrences of label $j$ in the vector $\bfa$. Then we have $\gamma_1+\cdots+\gamma_D=k$. For a fixed occupation vector \((\gamma_1,\ldots,\gamma_D)\), the number of
vectors \(\bfa\) with this occupation vector is
\begin{equation}
    \frac{k!}{\gamma_1!\cdots \gamma_D!}
\end{equation}
Hence the sum over pairs of vectors $(\bfa,\bfb)$ in \eref{eq:appFrameLHIndicator} can be rewritten as a sum over occupation vectors:
\begin{align}
    F^{(k)}_{\caE_{\scrL}(H)}
    &=
    \sum_{\substack{\gamma_1+\cdots+\gamma_D=k\\ \gamma_i\ge0}}
    \left(
        \frac{k!}{\gamma_1!\cdots \gamma_D!}
    \right)^2\notag\\
    &=
    k!D^k+\caO_k(D^{k-1}).
    \label{eq:appFrameLHOccupation}
\end{align}

\section{Proof of monotonicity of the frame potential under repeated perturbations}
\label{sec:app_frame_potential_monotonicity}

In this section, we prove that, for a fixed perturbation $U_0$, the frame potential of
$\caE^{(m)}_{\scrN}(H_1,U_0)$ is nonincreasing with the number of rounds of perturbation:
\begin{equation}
F_{\caE^{(m+1)}_{\scrN}(H_1,U_0)}^{(k)}
\leq
F_{\caE^{(m)}_{\scrN}(H_1,U_0)}^{(k)},
\end{equation}
as established in  \eref{eq:recursive_monotonicity} in the main text.
\begin{proof}Suppose $R_k(U):=U^{\otimes k}\otimes U^{*\otimes k}$. Then we have
\begin{equation}
\tr\left[R_k(U)^\dagger R_k(V)\right]
=
\left|\tr(U^\dagger V)\right|^{2k},\qquad F_{\caE}^{(k)}
=
\left\|
\bbE_{U\sim\caE}R_k(U)
\right\|_2^2.\label{eq:frame_potential_HS_form}
\end{equation}
For a fixed Hamiltonian $H_1$, define
\begin{equation}
M_m
:=
\bbE_{U\sim\caE_{\scrN}^{(m)}(H_1,U_0)}R_k(U),
\qquad
V_t:=U_0\rme^{-\rmi H_1t}.
\end{equation}
The definition of $\caE_{\scrN}^{(m)}$ yields
\begin{equation}
M_{m+1}
=\bbE_t\left(X_t\right)=
\bbE_t\left[M_mR_k(V_t)\right],
\label{eq:recursive_moment_matrix}
\end{equation}

where we have introduced the label $X_t:=M_mR_k(V_t)$. Then
\begin{equation}
\bbE_t\left\|X_t-M_{m+1}\right\|_2^2=
\bbE_t\|X_t\|_2^2
-2\operatorname{Re}\tr\left[
(\bbE_tX_t)^\dagger M_{m+1}
\right]
+\|M_{m+1}\|_2^2=
\bbE_t\|X_t\|_2^2
-\|M_{m+1}\|_2^2.
\label{eq:HS_variance_expansion}
\end{equation}
For every $t$, $V_t$ is unitary and hence $R_k(V_t)$ is unitary. Consequently,
\begin{equation}
\bbE_t\|X_t\|_2^2=\bbE_t\left\|M_mR_k(V_t)\right\|_2^2=
\|M_m\|_2^2.
\label{eq:HS_right_unitary_invariance}
\end{equation}
Combining \eref{eq:frame_potential_HS_form},
\eref{eq:HS_variance_expansion}, and
\eref{eq:HS_right_unitary_invariance} yields the exact decrement
\begin{equation}
F_{\caE_{\scrN}^{(m)}(H_1,U_0)}^{(k)}
-F_{\caE_{\scrN}^{(m+1)}(H_1,U_0)}^{(k)}=\|M_{m}\|_2^2-\|M_{m+1}\|_2^2=
\bbE_t\left\|
M_mR_k(V_t)-M_{m+1}
\right\|_2^2\geq0,
\label{eq:frame_potential_decrement}
\end{equation}
which proves the monotonicity of the frame potential under repeated perturbations.
\end{proof}

\section{Proof of \thref{thm:three_round_fixed_perturbation}}
\label{sec:proofthm2}

The proof follows the same time-averaging and Haar-twirling strategy as
the proof of \thref{thm:main_oneH_int}. We therefore reuse the
symmetry-factor counting and Weingarten estimates established in
\sref{sec:app_time_ave} and \sref{sec:app_prove_thm_main}, and focus
here on the new matching structure introduced by the second
perturbation.

\begin{proof}[Proof of \thref{thm:three_round_fixed_perturbation}]
We first prove the result for $m=3$. Let $U_0$ be a fixed unitary
satisfying the trace-suppression condition, and consider two independently
sampled unitaries
\begin{equation}
    U=\rme^{-\rmi H_1t_1}U_0\rme^{-\rmi H_1t_2}U_0\rme^{-\rmi H_1t_3},
    \qquad
    V=\rme^{-\rmi H_1t_1'}U_0\rme^{-\rmi H_1t_2'}U_0\rme^{-\rmi H_1t_3'}.
    \label{eq:app_three_round_UV}
\end{equation}
Following the time-averaging argument in
\pref{prp:time_averaged_permutation_sector},
the additive nonresonance condition implies that, in the perfect time-filter
limit,
\begin{equation}
    F_{\caE_{\scrN}^{(3)}(H_1,\{U_0\})}^{(k)}
    =
    \sum_{\pi,\sigma,\rho,\varphi\in S_k}
    \rmR_{\pi,\sigma,\rho,\varphi}(U_0),
    \label{eq:app_three_round_time_average}
\end{equation}
where
\begin{align}
    \rmR_{\pi,\sigma,\rho,\varphi}(U_0)
    &:=
    \sum_{\bfb,\bfc,\bfd,\mathbf e\in[D]^k}
    \frac{1}{\rmS(\bfb)\rmS(\bfc)\rmS(\bfd)\rmS(\mathbf e)}
    \prod_{r=1}^{k}
    \Big\{
    (U_0)_{b_rc_r}(U_0)_{c_rd_r}
    (U_0^\dagger)_{d_re_r}(U_0^\dagger)_{e_rb_r}
    \notag\\
    &\qquad\times
    \left[U_0^\dagger\right]_{c_{\sigma(r)}b_{\pi(r)}}
    \left[U_0^\dagger\right]_{d_{\rho(r)}c_{\sigma(r)}}
    \left[U_0\right]_{e_{\varphi(r)}d_{\rho(r)}}
    \left[U_0\right]_{b_{\pi(r)}e_{\varphi(r)}}
    \Big\}.
    \label{eq:app_three_round_R_def}
\end{align}

As in \aref{sec:app_time_ave}, the four symmetry factors remove the
overcounting caused by repeated entries in the four index vectors. We next
average over the Haar-random eigenbasis of $H_1$. To organize the
$8k$ matrix elements in \eref{eq:app_three_round_R_def}, introduce
\begin{equation}
    \Omega_3=
    \{\ell_r,m_r,n_r,o_r,p_r,q_r,s_r,t_r:r=1,\ldots,k\},
    \qquad |\Omega_3|=8k,
\end{equation}
and associate to each $\omega\in\Omega_3$ an operator $B^{(\omega)}$ and
indices $x_\omega,y_\omega$ according to
\[
\begin{array}{c|c|c|c}
\omega & B^{(\omega)} & x_\omega & y_\omega\\
\hline
\ell_r & U_0         & b_r                 & c_r\\
m_r    & U_0         & c_r                 & d_r\\
n_r    & U_0^\dagger & d_r                 & e_r\\
o_r    & U_0^\dagger & e_r                 & b_r\\
p_r    & U_0^\dagger & c_{\sigma(r)}       & b_{\pi(r)}\\
q_r    & U_0^\dagger & d_{\rho(r)}         & c_{\sigma(r)}\\
s_r    & U_0         & e_{\varphi(r)}      & d_{\rho(r)}\\
t_r    & U_0         & b_{\pi(r)}          & e_{\varphi(r)}
\end{array}
\]
Let $\bfx=(x_\omega)_{\omega\in\Omega_3}$ and
$\bfy=(y_\omega)_{\omega\in\Omega_3}$. The same Weingarten expansion as in
\coref{cor:app_frame_potential_expanded} gives
\begin{align}
    \overline{\rmR}_{\pi,\sigma,\rho,\varphi}
    &=
    \sum_{\alpha,\beta\in S_{8k}}
    \mathrm{Wg}^{(8k)}_D(\alpha^{-1}\beta)\,
    C_\alpha(U_0)\,
    \rmG_{\pi,\sigma,\rho,\varphi}(\beta),
    \label{eq:app_three_round_weingarten}\\
    C_\alpha(U_0)
    &:=
    \tr\left\{
    \left[
    \bigotimes_{\omega\in\Omega_3}B^{(\omega)}
    \right]V_\alpha^{\dagger}
    \right\},
    \label{eq:app_three_round_C_def}\\
    \rmG_{\pi,\sigma,\rho,\varphi}(\beta)
    &:=
    \sum_{\bfb,\bfc,\bfd,\mathbf e\in[D]^k}
    \frac{\<\bfx|V_\beta|\bfy\>}
    {\rmS(\bfb)\rmS(\bfc)\rmS(\bfd)\rmS(\mathbf e)}.
    \label{eq:app_three_round_G_def}
\end{align}
Here the overline denotes the Haar average over the eigenbasis of $H_1$.

We now identify the leading contributions in
\eref{eq:app_three_round_weingarten}. The argument parallels
\pref{prp:N_beta_upperbound} and \pref{prp:C_leading_order}, but we
make the corresponding power counting explicit. Since $k$ is fixed,
the cardinality of $S_{8k}$ is independent of $D$, so it is sufficient
to estimate each term in the Weingarten expansion separately.

First consider $C_\alpha(U_0)$ defined in
\eref{eq:app_three_round_C_def}. We first recall the traceless case.
If $\tr(U_0)=0$, any fixed point of $\alpha$ produces a factor
$\tr(U_0)$ or $\tr(U_0^\dagger)$ and hence gives a vanishing
contribution. If $\alpha$ has no fixed points but contains a cycle of
length greater than two, then
$|\mathrm{cyc}(\alpha)|\le 4k-1$. Since each cycle contributes in
magnitude at most $D$, it follows that
\begin{equation}
    |C_\alpha(U_0)|
    \le D^{4k-1}
    =
    \caO_k(D^{4k-1}).
    \label{eq:app_three_round_C_suppressed}
\end{equation}
Thus, in the traceless case, an $\caO_k(D^{4k})$ contribution is
possible only when $\alpha$ consists entirely of transpositions.

We now show that the same leading-sector restriction holds for a
general fixed trace-suppressed perturbation. We have $|\tr(U_0)|^2=\epsilon\ll D$ as implied by the trace-suppression condition. Suppose that
$\alpha\in S_{8k}$ contains $s\ge1$ fixed points. Each fixed point
contributes a factor $\tr(U_0)$ or $\tr(U_0^\dagger)$, while the
remaining $8k-s$ labels belong to cycles of length at least two.
Hence,
\begin{align}
    |C_\alpha(U_0)|
    &\le
    |\tr(U_0)|^s
    D^{\lfloor(8k-s)/2\rfloor}
    \notag\\
    &\le
    \left(\frac{\epsilon}{D}\right)^{s/2}
    D^{4k}
    \ll
    D^{4k}.
    \label{eq:app_three_round_C_fixedpoint_suppressed}
\end{align}
Thus any sector containing at least one fixed point is asymptotically
subleading and cannot contribute to the leading
$\caO_k(D^{4k})$ sector. If instead
$s=0$ but $\alpha$ contains a cycle of length greater than two, the same
counting as in the traceless case gives
$|C_\alpha(U_0)|=\caO_k(D^{4k-1})$. Consequently, an
$\caO_k(D^{4k})$ contribution is possible only when
$\alpha\in M_{8k}$, where $M_{8k}\subset S_{8k}$ denotes the set of
permutations consisting of $4k$ disjoint transpositions.

We next consider
$\rmG_{\pi,\sigma,\rho,\varphi}(\beta)$ defined in
\eref{eq:app_three_round_G_def}. As in the proof of
\pref{prp:N_beta_upperbound}, the matrix element
$\langle\bfx|V_\beta|\bfy\rangle$ is a product of Kronecker delta
functions imposing equality constraints among the $4k$ variables
contained in $\bfb,\bfc,\bfd$, and $\mathbf e$. If $V_\beta$ does not
perfectly match the indices of $\bfx$ and $\bfy$, at least one
nontrivial equality constraint is imposed, so at most $4k-1$
independent indices remain. Therefore,
\begin{equation}
    \rmG_{\pi,\sigma,\rho,\varphi}(\beta)
    =
    \caO_k(D^{4k-1}).
    \label{eq:app_three_round_G_suppressed}
\end{equation}
If the matching is perfect, the matrix element is identically one and
\pref{prp:A_kD_upp} gives
\begin{align}
    \rmG_{\pi,\sigma,\rho,\varphi}(\beta)
    &=
    \left[
        \sum_{\bfb\in[D]^k}\frac{1}{\rmS(\bfb)}
    \right]^4
    =
    [A_k(D)]^4
    =
    D^{4k}+\caO_k(D^{4k-1}).
    \label{eq:app_three_round_G_leading}
\end{align}

Together with the Weingarten scaling
\eref{eq:wg_scaling_schur}, these estimates completely determine the
leading sector of \eref{eq:app_three_round_weingarten}. Indeed, if
$\alpha\notin M_{8k}$, the bounds
\eref{eq:app_three_round_C_suppressed} and
\eref{eq:app_three_round_C_fixedpoint_suppressed} show that
$C_\alpha(U_0)=o_{D\to\infty}(D^{4k})$, so this sector is strictly
subleading after the Weingarten normalization. Likewise, an imperfect
index matching is suppressed by
\eref{eq:app_three_round_G_suppressed}. Finally, when
$\alpha\in M_{8k}$ and the matching is perfect but
$\alpha\neq\beta$,
\begin{equation}
    \mathrm{Wg}^{(8k)}_D(\alpha^{-1}\beta)
    =
    \caO_k\!\left(
        D^{-8k-|\alpha^{-1}\beta|}
    \right)
    =
    \caO_k(D^{-8k-1}),
\end{equation}
so this contribution is again $o_{D\to\infty}(1)$. Since the number
of permutations is finite for fixed $k$, the sum of all such
suppressed terms remains $o_{D\to\infty}(1)$. Hence an
$\caO_k(1)$ contribution requires
\[
    \alpha=\beta\in M_{8k}
\]
together with a perfect index matching.

It remains to classify these perfect transposition matchings. The
essential difference from the proof of
\thref{thm:main_oneH_int} is that an internal contraction among the
first four rows of the index-incidence table propagates, through their
shared indices, into the cycle
\begin{equation}
    \ell_r
    \longrightarrow
    m_r
    \longrightarrow
    n_r
    \longrightarrow
    o_r
    \longrightarrow
    \ell_r.
    \label{eq:app_three_round_direct_cycle}
\end{equation}
This is a four-cycle and is therefore incompatible with
$\alpha\in M_{8k}$. Thus no such internal contraction can occur in a
leading perfect matching, and all four labels associated with each
replica must be paired with labels in the last four rows.

The allowed crossed transpositions and their consistency conditions
are therefore
\begin{align}
    \ell_r&\leftrightarrow p_{\pi^{-1}(r)},
    &\sigma\pi^{-1}(r)&=r,
    \notag\\
    m_r&\leftrightarrow q_{\sigma^{-1}(r)},
    &\rho\sigma^{-1}(r)&=r,
    \notag\\
    n_r&\leftrightarrow s_{\rho^{-1}(r)},
    &\varphi\rho^{-1}(r)&=r,
    \notag\\
    o_r&\leftrightarrow t_{\varphi^{-1}(r)},
    &\pi\varphi^{-1}(r)&=r.
    \label{eq:app_three_round_crossed_matching}
\end{align}
For these relations to hold for every $r$, one must have
\begin{equation}
    \pi=\sigma=\rho=\varphi.
    \label{eq:app_three_round_equal_permutations}
\end{equation}
Conversely, for each $\pi\in S_k$,
\eref{eq:app_three_round_equal_permutations} gives exactly one perfect
crossed matching,
\begin{equation}
    \alpha_\pi=
    \prod_{r=1}^{k}
    (\ell_r\ p_{\pi^{-1}(r)})
    (m_r\ q_{\pi^{-1}(r)})
    (n_r\ s_{\pi^{-1}(r)})
    (o_r\ t_{\pi^{-1}(r)}).
    \label{eq:app_three_round_alpha_matching}
\end{equation}

For this matching, every transposition contracts $U_0$ with
$U_0^\dagger$, and hence
\begin{equation}
    C_{\alpha_\pi}(U_0)=D^{4k},
    \qquad
    \mathrm{Wg}^{(8k)}_D(e)
    =
    D^{-8k}+\caO_k(D^{-8k-2}).
\end{equation}
Combining this result with
\eref{eq:app_three_round_G_leading}, we obtain
\begin{equation}
    \overline{\rmR}_{\pi,\sigma,\rho,\varphi}
    =
    \delta_{\pi,\sigma}
    \delta_{\sigma,\rho}
    \delta_{\rho,\varphi}
    +
    o_{D\to\infty}(1).
    \label{eq:app_three_round_R_leading}
\end{equation}
There are exactly $k!$ choices satisfying
$\pi=\sigma=\rho=\varphi$. Therefore,
\begin{equation}
    \overline{F}_{\caE_{\scrN}^{(3)}(H_1,U_0)}^{(k)}
    =
    \sum_{\pi,\sigma,\rho,\varphi\in S_k}
    \overline{\rmR}_{\pi,\sigma,\rho,\varphi}
    =
    k!+o_{D\to\infty}(1).
    \label{eq:app_three_round_m3_result}
\end{equation}

Finally, the frame potential is bounded below by its Haar value, while
\sref{sec:app_frame_potential_monotonicity} shows that it is
nonincreasing under additional perturbation rounds. Hence, for every
$m\ge3$,
\begin{equation}
    k!
    \le
    \overline{F}_{\caE_{\scrN}^{(m)}(H_1,U_0)}^{(k)}
    \le
    \overline{F}_{\caE_{\scrN}^{(3)}(H_1,U_0)}^{(k)}
    =
    k!+o_{D\to\infty}(1),
\end{equation}
which proves
\thref{thm:three_round_fixed_perturbation}.
\end{proof}

\section{Proof of \thref{thm:local_space_depth_bound}}
\label{sec:proof_local_space_depth}

As in the main text, we choose the
energy origin such that $\tr(H)=0$ throughout the resource analysis. Write $h:=H/\sigma_H$, where
$\sigma_H^2=\tr(H^2)/D$, so that $\tr(h)=0$ and $\tr(h^2)/D=1$. For any operator $B$ on $\mathcal H$, we define its insertion on replica $j$ as
\begin{equation}
    B_j:=I_D^{\otimes(j-1)}\otimes B\otimes I_D^{\otimes(k-j)},
    \qquad j=1,\ldots,k.
    \label{eq:sdproof_replica_insertion}
\end{equation}
In particular, $H_j$ and $h_j$ act on an entire $n$-qubit replica,
instead of on physical site $j$. The generator of replicated free evolution is
\begin{equation}
    H^{[k]}:=\sum_{j=1}^kH_j,
    \qquad (e^{-iHt})^{\otimes k}=e^{-itH^{[k]}}.
    \label{eq:sdproof_replica_generator}
\end{equation}

The proof of \thref{thm:local_space_depth_bound} relies on the following two lemmas.

\begin{lemma}
\label{lem:sdproof_energy}
Let $h$ be a fixed Hamiltonian, and suppose that the perturbations in
the $m$ rounds are sampled independently and identically from an
ensemble $\caE_{\mathrm{int}}$ of Hermitian unitaries. Then, for $m\ge1$ and $D\ge k+1$,
\begin{equation}
    F_{m+1,\infty}^{(k)}-k!
    \ge
    k\,k!\,
    \left|
    \mathbb E_{U\sim\caE_{\mathrm{int}}}
    \frac{\tr(hUhU^\dagger)}{D}
    \right|^{2m}.
    \label{eq:sdproof_energy_bound}
\end{equation}
\end{lemma}

\begin{lemma}
\label{lem:sdproof_locality}
Under the  assumptions of \thref{thm:local_space_depth_bound} in \eqsref{eq:local_squared_weight_bound}{eq:extensive_energy_variance_bound}, we have
\begin{equation}
    \mathbb E_{U\sim\caE_{\mathrm{int}}}\frac{\tr(hUhU^\dagger)}{D}\ge1-\frac{2J_{\mathrm{loc}}^2\bar w}{vn}.
    \label{eq:sdproof_local_c}
\end{equation}
\end{lemma}

\begin{proof}[Proof of \lref{lem:sdproof_energy}]
Let $\mathcal R_k(V)(X):=V^{\otimes k}X(V^\dagger)^{\otimes k}$, and
introduce the time average, perturbation average, and Haar-twirling:
\begin{equation}
    \Pi_k:=\lim_{T\to\infty}\frac1T\int_0^T
           \mathcal R_k(e^{-iHt})\,dt,
    \qquad K_k:=\mathbb E_{U\sim\caE_{\mathrm{int}}}\mathcal R_k(U),
    \qquad Q_k:=\int_{\mathrm U(D)}\mathcal R_k(V)\,dV.
    \label{eq:sdproof_maps}
\end{equation}
The time average is the orthogonal projection onto operators commuting
with $H^{[k]}$, including all degenerate energy blocks. The Haar-twirling
is also an orthogonal projection, and
$\Pi_kQ_k=Q_k\Pi_k=K_kQ_k=Q_kK_k=Q_k$.
Since the perturbations are Hermitian, $K_k$ is self-adjoint. Therefore
\begin{equation}
    A_k:=\Pi_kK_k\Pi_k-Q_k
    \quad\text{is self-adjoint},\qquad
    F_{m+1,\infty}^{(k)}-k!=\|A_k^m\|_{\mathrm F}^2.
    \label{eq:sdproof_moment_identity}
\end{equation}
Here $\|\cdot\|_{\mathrm F}$ denotes the Frobenius norm of a
superoperator. Independence gives the moment map
$(\Pi_kK_k\Pi_k)^m=Q_k+A_k^m$, whose squared Frobenius norm equals
the frame potential. Since $Q_kA_k=A_kQ_k=0$ and
$\|Q_k\|_{\mathrm F}^2=k!$ for $D\ge k$, subtracting the Haar
contribution gives the second identity.

We now construct enough operators to lower-bound the norm $\|A_k^m\|$. Let $h_j$
denote $h$ on replica $j$, and let $W_\pi$ permute replicas according to
$W_\pi|i_1,\ldots,i_k\rangle
=|i_{\pi^{-1}(1)},\ldots,i_{\pi^{-1}(k)}\rangle$. Consider
\begin{equation}
    X_{j,\pi}:=h_jW_\pi,
    \qquad j=1,\ldots,k,\quad\pi\in S_k.
    \label{eq:sdproof_modes}
\end{equation}
This simply places the conserved energy on one replica and combines it
with a replica permutation. Both factors commute with $H^{[k]}$,
so $\Pi_k(X_{j,\pi})=X_{j,\pi}$. Moreover,
\begin{equation}
    Q_k(X_{j,\pi})
      =\left(\int VhV^\dagger\,dV\right)_jW_\pi=0,
    \label{eq:sdproof_haar_zero}
\end{equation}
because $\tr(h)=0$.

These $N:=k\,k!$ operators are linearly independent when $D\ge k+1$.
To see this, work in an eigenbasis $h|i\rangle=e_i|i\rangle$ and suppose
$\sum_{j,\pi}z_{j,\pi}h_jW_\pi=0$. On product states with distinct
labels, different permutations give orthogonal outputs. Relabeling the
output indices therefore gives, for each $\pi$,
\begin{equation}
    \sum_{j=1}^kz_{j,\pi}e_{i_j}=0
    \qquad\text{whenever }i_1,\ldots,i_k\text{ are distinct}.
    \label{eq:sdproof_independence}
\end{equation}
Choose labels $p,q$ with $e_p\ne e_q$, which exist because $h$ is
nonzero and traceless. For each $j$, keep the other $k-1$ labels fixed
outside $\{p,q\}$ and compare $i_j=p$ with $i_j=q$;
this is possible for $D\ge k+1$.
Subtracting the two equations gives $z_{j,\pi}(e_p-e_q)=0$, proving
independence without any nondegeneracy assumption.

The operators are not generally orthogonal. Their Gram matrix is
\begin{equation}
    G_{(j,\pi),(\ell,\sigma)}
    :=D^{-k}\tr(X_{j,\pi}^\dagger X_{\ell,\sigma})
    =D^{-k}\tr(h_jh_\ell W_{\sigma\pi^{-1}}),
    \label{eq:sdproof_gram_definition}
\end{equation}
where the second equality uses cyclicity of the trace. Here $\pi,\sigma$
are permutations of replicas; $\sigma$ is unrelated to the energy scale
$\sigma_H$. Set $\tau=\sigma\pi^{-1}$ and let $b(\tau)$ be its number
of cycles, including fixed points. In the same eigenbasis of $h$,
\begin{equation}
    \tr(h_jh_\ell W_\tau)
    =\sum_{i_1,\ldots,i_k=1}^{D}e_{i_j}e_{i_\ell}
       \prod_{s=1}^{k}\delta_{i_s,i_{\tau^{-1}(s)}}.
    \label{eq:sdproof_cycle_indices}
\end{equation}
The Kronecker deltas require all indices along each cycle to coincide,
leaving one independently summed index per cycle. A cycle containing
neither insertion contributes a factor of $D$. If $j$ and $\ell$ lie on
the same cycle, that cycle contributes $\sum_i e_i^2=\tr(h^2)=D$;
the other $b(\tau)-1$ cycles contribute $D^{b(\tau)-1}$.
The trace is therefore $D^{b(\tau)}$. If $j$ and $\ell$ lie on
different cycles, each of those two cycles instead contributes
$\sum_i e_i=\tr(h)=0$, and the trace vanishes. Dividing by $D^k$ gives
\begin{equation}
    G_{(j,\pi),(\ell,\sigma)}
      =D^{b(\tau)-k}
       \begin{cases}
       1,&j,\ell\text{ lie in the same cycle of }\tau,\\
       0,&\text{otherwise}.
       \end{cases}
    \label{eq:sdproof_gram}
\end{equation}
The same-cycle case includes $j=\ell$, when both insertions form
$h_j^2$. For example, at $k=2$, $G_{(1,\mathrm{id}),(2,\mathrm{id})}
=D^{-2}\left[\tr (h)\right]^2=0$, whereas
$G_{(1,\mathrm{id}),(2,(12))}=D^{-2}\tr(h^2)=D^{-1}$.

Let $g:=\mathbb E_{U\sim\caE_{\mathrm{int}}}(UhU^\dagger)$, so that $\tr(g)=0$. Since $K_k(h_\ell W_\sigma)=g_\ell W_\sigma$, the same
cycle calculation yields
\begin{equation}
\begin{aligned}
    D^{-k}\tr\!\left[X_{j,\pi}^\dagger A_k(X_{\ell,\sigma})\right]
      &=D^{-k}\tr(h_jg_\ell W_\tau)\\
      &=\left[\mathbb E_{U\sim\caE_{\mathrm{int}}}\frac{\tr(hUhU^\dagger)}{D}\right]G_{(j,\pi),(\ell,\sigma)}.
\end{aligned}
    \label{eq:sdproof_compression}
\end{equation}

Independence makes $G$ positive definite. With $\alpha$ indexing the
pairs $(j,\pi)$, orthogonalize by setting
$Y_\mu:=\sum_\alpha X_\alpha(G^{-1/2})_{\alpha\mu}$,
$\mu=1,\ldots,N$. These are the $N$ specific operators obtained by
orthogonalizing the constructed family. 
\Eqsref{eq:sdproof_gram}{eq:sdproof_compression} give
\begin{equation}
\begin{aligned}
    D^{-k}\tr(Y_\mu^\dagger Y_\nu)&=\delta_{\mu\nu},\\
    D^{-k}\tr\!\left[Y_\mu^\dagger A_k(Y_\nu)\right]&=\left[\mathbb E_{U\sim\caE_{\mathrm{int}}}\frac{\tr(hUhU^\dagger)}{D}\right]\delta_{\mu\nu}.
\end{aligned}
    \label{eq:sdproof_orthogonal}
\end{equation}
These are overlap identities, not an assertion that the $Y_\mu$ are
eigenoperators. To control repeated rounds, let $Y$ denote any operator
on $\mathcal H^{\otimes k}$ satisfying $D^{-k}\tr(Y^\dagger Y)=1$;
we will apply the resulting inequality to each $Y=Y_\mu$.
Self-adjointness of $A_k$ gives a complete eigenoperator basis
$\{Z_b\}_{b=1}^{D^{2k}}$ such that
\begin{equation}
    A_k(Z_b)=\lambda_b Z_b,\qquad \lambda_b\in\mathbb R,
    \qquad D^{-k}\tr(Z_b^\dagger Z_{b'})=\delta_{bb'}.
    \label{eq:sdproof_spectral_basis}
\end{equation}
It is the superoperator $A_k$, rather than the physical Hamiltonian,
that is diagonalized here. Explicitly,
\begin{equation}
\begin{aligned}
    Y&=\sum_b y_bZ_b,\qquad
      y_b=D^{-k}\tr(Z_b^\dagger Y),\\
    A_k^r(Y)&=\sum_b y_b\lambda_b^r Z_b,\qquad
      p_b:=|y_b|^2\ge0,\quad \sum_b p_b=1.
\end{aligned}
    \label{eq:sdproof_spectral_expansion}
\end{equation}
The last equality follows from the normalization of $Y$ and
orthonormality of the $Z_b$. In particular,
$D^{-k}\tr[Y^\dagger A_k(Y)]=\sum_b p_b\lambda_b$.
Jensen's inequality then gives
\begin{equation}
    D^{-k}\|A_k^m(Y)\|_2^2
      =\sum_bp_b\lambda_b^{2m}
      \ge\left(\sum_bp_b\lambda_b\right)^{2m}
      =\left|D^{-k}\tr\!\left[Y^\dagger A_k(Y)\right]\right|^{2m}.
    \label{eq:sdproof_jensen}
\end{equation}
The weights $p_b$ depend on $Y$; when $Y=Y_\mu$, they may differ for
different $\mu$, but the one-round overlap is the same by
\eref{eq:sdproof_orthogonal}. Extend the $N$ operators $Y_\mu$ to a
complete basis normalized by $D^{-k}\tr(E_b^\dagger E_{b'})=\delta_{bb'}$.
Since $\|A_k^m\|_{\mathrm F}^2=D^{-k}\sum_b\|A_k^m(E_b)\|_2^2$,
retaining only the $N$ nonnegative contributions from $Y_\mu$ and using
\eref{eq:sdproof_moment_identity} proves
\begin{equation}
    F_{m+1,\infty}^{(k)}-k!
      \ge D^{-k}\sum_{\mu=1}^{k\,k!}\|A_k^m(Y_\mu)\|_2^2
      \ge k\,k!\,\left|\mathbb E_U\frac{\tr(hUhU^\dagger)}{D}\right|^{2m}.
\end{equation}
\end{proof}

\begin{proof}[Proof of \lref{lem:sdproof_locality}]
For $S=\operatorname{supp}(U)$, let $H_S$ be the sum of the Pauli
components of $H$ whose supports intersect $S$, and write
\begin{equation}
    W_H(S):=D^{-1}\|H_S\|_2^2
      =\sum_{\alpha:\,\operatorname{supp}(P_\alpha)\cap S
                      \ne\varnothing}|h_\alpha|^2
      \le J_{\mathrm{loc}}^2|S|.
\end{equation}
The last inequality follows by summing the per-site weight bound over
$S$, which may count some terms more than once. All terms in $H-H_S$
commute with $U$. The triangle inequality and unitary invariance imply
\begin{equation}
    \|UHU^\dagger-H\|_2^2
      =\|UH_SU^\dagger-H_S\|_2^2
      \le4D W_H(S).
\end{equation}
Expanding this squared difference, averaging over $U$, and using
$\sigma_H^2\ge vn$ gives
\begin{equation}
    \begin{aligned}
    \mathbb E_U\frac{\tr(hUhU^\dagger)}{D}&=1-\frac{\mathbb E_U\|UHU^\dagger-H\|_2^2}
                    {2D\sigma_H^2}\ge1-\frac{2J_{\mathrm{loc}}^2\bar w}{\sigma_H^2}
      \ge1-\frac{2J_{\mathrm{loc}}^2\bar w}{vn}.
    \end{aligned}
\end{equation}
\end{proof}

The proof of \thref{thm:local_space_depth_bound} then follows naturally by combining \Lsref{lem:sdproof_energy}{lem:sdproof_locality}.
\begin{proof}[Proof of \thref{thm:local_space_depth_bound}]
Put $x:=2J_{\mathrm{loc}}^2\bar w/(vn)$. The support assumption gives
$0<x\le1/2$, so Lemma~\ref{lem:sdproof_locality} implies
$\mathbb E_U\frac{\tr(hUhU^\dagger)}{D}\ge1-x>0$. Lemma~\ref{lem:sdproof_energy} then yields
\begin{equation}
    \frac{F_{m+1,\infty}^{(k)}-k!}{k!}
      \ge k(1-x)^{2m}
      =k\left(1-\frac{2J_{\mathrm{loc}}^2\bar w}{vn}\right)^{2m}.
\end{equation}
If the relative excess is at most $\varepsilon<k$, necessarily
\begin{equation}
    m\ge\frac{\log(k/\varepsilon)}{-2\log(1-x)}
      \ge\frac{\log(k/\varepsilon)}{4x},
\end{equation}
where $-\log(1-x)\le x/(1-x)\le2x$ for $0<x\le1/2$.
Multiplying by $\bar w$ gives the claimed area bound,
\begin{equation}
    \mathcal A_{\mathrm{ctrl}}=m\bar w
       \ge\frac{vn}{8J_{\mathrm{loc}}^2}\,
                  \log\left(\frac{k}{\varepsilon}\right),
\end{equation}
which proves \thref{thm:local_space_depth_bound}.
\end{proof}

\section{Numerical sampling procedure}
\label{sec:numerical_sampling}

Frame potentials were estimated by Monte Carlo sampling independent pairs
of unitaries $(U,V)$ from the corresponding temporal ensemble and averaging
$|\tr(U^\dagger V)|^{2k}$. The samples were partitioned into independent
batches; the frame potential was first estimated within each batch, and the
reported value was obtained by averaging over the batch estimates. For
several intermediate ensembles, such as the nontrivial Pauli ensemble, rare
events in which the two independently sampled intermediate perturbations
coincide can produce exceptionally large contributions to
$|\tr(U^\dagger V)|^{2k}$, particularly at large $k$, and therefore lead to
large fluctuations under direct sampling. We treat this sector separately
by estimating the conditional averages for coincident and noncoincident
intermediate perturbations and recombining them according to their exact
sampling probabilities. This stratified procedure preserves the exact
expectation value while substantially reducing the Monte Carlo variance.

For the random-matrix calculations, the GUE Hamiltonians were generated by
sampling the matrix elements with the standard $D$-dependent normalization,
such that $\sigma_H^2=\tr(H^2)/D=\caO(1)$ as $D$ increases. For the Rydberg
calculations, each reported frame-potential data point was additionally
averaged over 20 independent Hamiltonian realizations.

The spectral statistics of the Rydberg Hamiltonian were obtained by
diagonalizing $H_{\mathrm{Ryd}}$ for each realization and ordering the
eigenvalues as $E_1<E_2<\cdots<E_D$. The nearest-neighbor level spacings
$\Delta_j=E_{j+1}-E_j$ were normalized separately for each realization by
their mean value, $s_j=\Delta_j/\overline{\Delta}$. The normalized spacings
from 20 independent Hamiltonian realizations were then combined to estimate
the level-spacing distribution $P(s)$. The numerical distribution was
compared with the GOE Wigner--Dyson prediction
$P_{\mathrm{GOE}}(s)=(\pi/2)s\exp(-\pi s^2/4)$ and the Poisson prediction
$P_{\mathrm{P}}(s)=\exp(-s)$.

\clearpage
\section{Additional numerical results on finite-time convergence}
\label{sec:finite_time_numerics}

To complement the finite-window analysis in the main text, we provide
additional numerical results showing how the frame potential approaches
its long-time value for different intermediate ensembles. These results
demonstrate that the design-forming behavior predicted by
\thref{thm:main_oneH_int} already emerges at rather moderate sampling windows,
well before the formal perfect time-filter limit $T\to\infty$.

\begin{figure}[!htbp]
    \centering
    \includegraphics[width=0.8\textwidth]{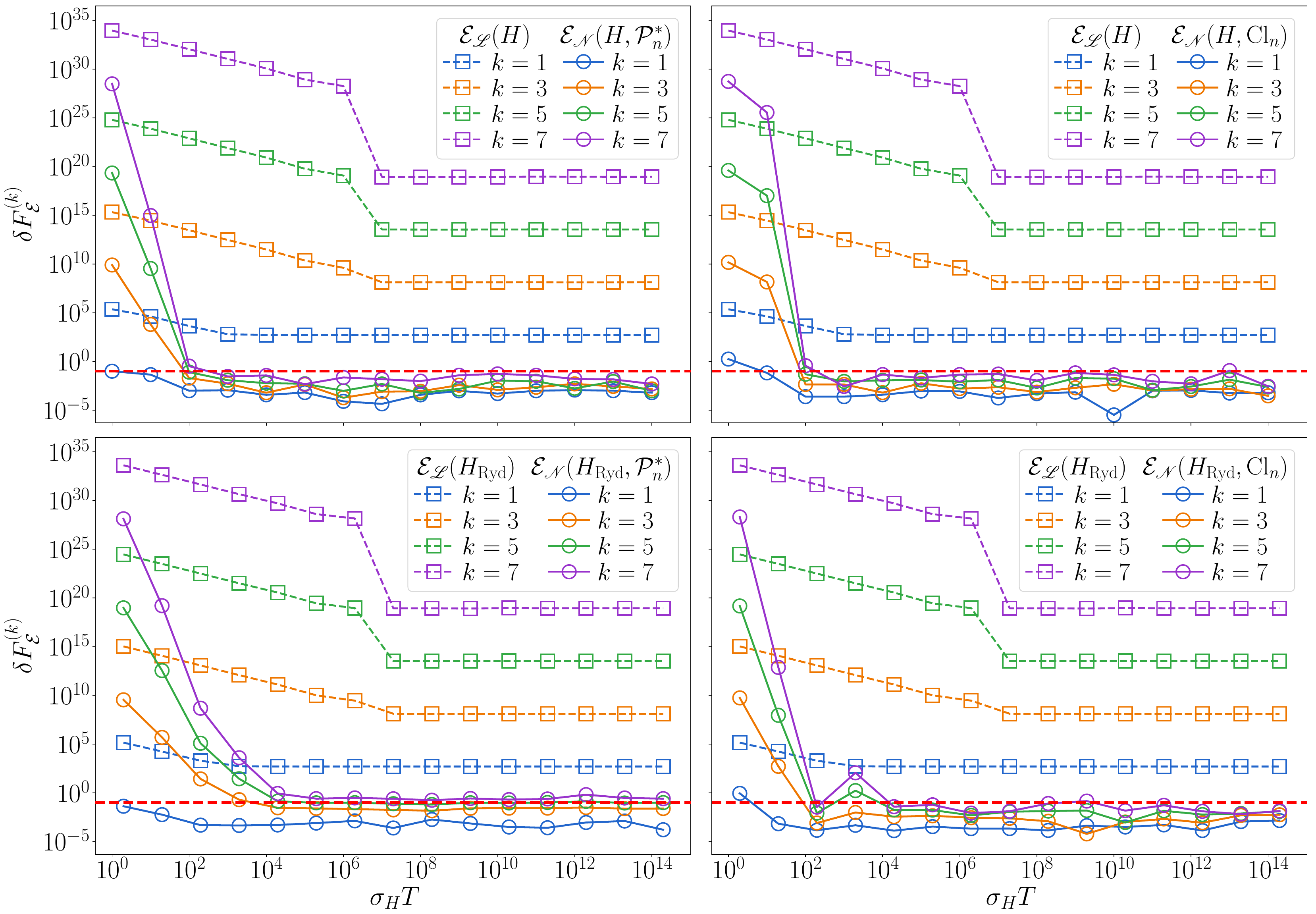}
    \caption{
    The normalized deviation of the frame potential
    $\delta F_{\mathcal E}^{(k)}(T)$ as a function of $\sigma_HT$ for the relevant
    ensembles. Here, the chaotic Hamiltonian $H$ in the upper row acts on
    $n=9$ qubits; after being sampled from the GUE, it is fixed throughout
    the protocol. The chaotic Hamiltonian $H_{\mathrm{Ryd}}$ in the lower row is the Rydberg
    Hamiltonian in \eref{eq:rydHamiltonian}, acting on $n=9$ qubits, with
    $\Omega_i=\Omega_0(1+\epsilon_i)$ and
    $\Delta_i=\Delta_0(1+\eta_i)$. The inhomogeneities $\epsilon_i$ and
    $\eta_i$ are sampled independently and uniformly from $[-0.12,0.12]$
    and $[-0.35,0.35]$, respectively. For the GUE Hamiltonian $H$, each
    point is averaged over $10^7$ independent samples. For each realization
    of $H_{\mathrm{Ryd}}$, every data point is averaged over $10^7$ random
    samples of the corresponding unitary ensemble; the final data points
    are further averaged over 20 independent realizations of
    $H_{\mathrm{Ryd}}$. The horizontal red dashed line marks the reference
    threshold $\delta F=0.1$.
    }
    \label{fig:finite_time_deltaF}
\end{figure}

\Fref{fig:finite_time_deltaF} compares the unperturbed single-Hamiltonian
temporal ensemble $\mathcal E_{\scrL}(H)$ with the perturbed protocols
$\mathcal E_{\scrN}(H,\caP_n^*)$ and
$\mathcal E_{\scrN}(H,\Cl_n)$, for both a fixed GUE Hamiltonian and the
Rydberg Hamiltonian defined in \eref{eq:rydHamiltonian}.

For the unperturbed temporal ensemble $\mathcal E_{\scrL}(H)$, the frame
potential converges slowly with $T$ and saturates far above the Haar value.
In contrast, after introducing an intermediate perturbation,
$\delta F_{\mathcal E}^{(k)}$ decreases much more rapidly for both
$\caP_n^*$ and $\Cl_n$, and for both choices of chaotic Hamiltonian,
reaching the Haar-level regime within a substantially shorter sampling
window. These results provide finite-time numerical support for the
design-forming mechanism established in \thref{thm:main_oneH_int} and
complement the depth-dependent finite-window scaling discussed in the main
text.

In particular, the results show that the asymptotic design-forming behavior
does not require access to the perfect time-filter limit in the systems
studied here: finite sampling windows together with suitable intermediate
perturbations are already sufficient to produce near-Haar frame potentials.

\end{document}